\documentclass[12pt]{article}
\usepackage{amsmath,amssymb}
\usepackage{natbib}
\usepackage{amsthm}

\numberwithin{equation}{section}

\usepackage{rotating}
\usepackage{multirow}

\bibpunct{(}{)}{;}{a}{}{,}

\newcommand{\matern}{Mat\'{e}rn}

\newcommand{\ba}{\begin{array}}
\newcommand{\ea}{\end{array}}
\newcommand{\bea}{\begin{eqnarray}}
\newcommand{\eea}{\end{eqnarray}}
\newcommand{\beas}{\begin{eqnarray*}}
\newcommand{\eeas}{\end{eqnarray*}}
\newcommand{\be}{\begin{equation}}
\newcommand{\ee}{\end{equation}}
\newcommand{\bes}{\[}
\newcommand{\ees}{\]}
\newcommand{\bc}{\begin{center}}
\newcommand{\ec}{\end{center}}
\newcommand{\ben}{\begin{enumerate}}
\newcommand{\een}{\end{enumerate}}

\newcommand{\bi}{\begin{itemize}}
\newcommand{\ei}{\end{itemize}}

\newenvironment{BraceEquationArray}{
                          \begin{eqnarray}
                          \left\{
                          \begin{array}{rclr}
                         }{
                          \end{array}
                          \right.
                          \end{eqnarray}
                         }
\newenvironment{BraceEquationArrayS}{
                          \begin{eqnarray*}
                          \left\{
                          \begin{array}{rclr}
                         }{
                          \end{array}
                          \right.
                          \end{eqnarray*}
                        }
\newenvironment{EquationArray}{
                          \begin{eqnarray}
                          \begin{array}{rclr}
                         }{
                          \end{array}
                          \end{eqnarray}
                         }
\newenvironment{EquationArrayS}{
                          \begin{eqnarray*}
                          \begin{array}{rclr}
                         }{
                          \end{array}
                          \end{eqnarray*}
                         }

\newcommand{\bBEA}{\begin{BraceEquationArray}}
\newcommand{\eBEA}{\end{BraceEquationArray}}
\newcommand{\bBEAS}{\begin{BraceEquationArrayS}}
\newcommand{\eBEAS}{\end{BraceEquationArrayS}}
\newcommand{\bEA}{\begin{EquationArray}}
\newcommand{\eEA}{\end{EquationArray}}
\newcommand{\bEAs}{\begin{EquationArrayS}}
\newcommand{\eEAs}{\end{EquationArrayS}}

\newcommand{\ve}{\varepsilon}
\newcommand{\goto}{\rightarrow}

\newcommand{\brmk}{\begin{remark}\begin{em}}
\newcommand{\ermk}{\end{em}\end{remark}}
\newcommand{\bexa}{\begin{example}\per\begin{em}}
\newcommand{\eexa}{\end{em}\end{example}}

\newcommand{\R}{\mathbb R}

\newcommand{\V}{{\rm Var}}
\newcommand{\C}{{\rm Cov }}

\newcommand{\bbeta}{{\boldsymbol\beta}}

\newcommand{\bgamma}{{\boldsymbol\gamma}}

\newcommand{\bxi}{{\boldsymbol\xi}}
\newcommand{\bomega}{{\boldsymbol\omega}}
\newcommand{\bzeta}{{\boldsymbol\zeta}}

\newcommand{\bTheta}{{\mathbf{\Theta}}}
\newcommand{\bPhi}{{\mathbf{\Phi}}}
\newcommand{\bpsi}{{\boldsymbol\psi}}

\newcommand{\s}{{\mathbf{s}}}
\newcommand{\z}{{\mathbf{z}}}

\newcommand{\capx}{{\mathbf{X}}}
\newcommand{\capr}{{\mathbf{R}}}
\newcommand{\capa}{{\mathbf{A}}}

\newcommand{\capm}{{\mathbf{M}}}

\newcommand{\rtmri}{\big(\capr^{*^ \top} \capm \capr^*\big)^{-1}}
\newcommand{\rtdmjr}{\capr^{*^ \top}\frac{\partial \capm}{\partial m_j}\capr^*}
\newcommand{\rtdmkr}{\capr^{*^ \top}\frac{\partial \capm}{\partial m_k}\capr^*}
\newcommand{\rtmp}{\capr^{*^ \top}\capm \bPhi^*}
\newcommand{\rtddmr}{\capr^{*^ \top}\frac{\partial^2 \capm}{\partial m_j \partial m_k}\capr^*}
\newcommand{\rtdmjp}{\capr^{*^ \top}\frac{\partial \capm}{\partial m_j}\bPhi^*}
\newcommand{\rtdmkp}{\capr^{*^ \top}\frac{\partial \capm}{\partial m_k}\bPhi^*}
\newcommand{\rtddmp}{\capr^{*^ \top}\frac{\partial^2 \capm}{\partial m_j \partial m_k}\bPhi^*}

\DeclareMathOperator*{\argmin}{argmin}

\newtheorem{lem}{Lemma}
\newtheorem{cond}{Condition}
\newtheorem{defin}{Definition}

\begin{document}

\title{Measurement error in two-stage analyses, with application to air pollution epidemiology }
\author{Adam A. Szpiro and Christopher J. Paciorek}

\maketitle

\begin{abstract}

Public health researchers often estimate health effects of exposures (e.g., pollution, diet, lifestyle) that cannot be directly measured for study subjects. A common strategy in environmental epidemiology is to use a first-stage (exposure) model to estimate the exposure based on covariates and/or spatio-temporal proximity and to use predictions from the exposure model as the covariate of interest in the second-stage (health) model. This induces a complex form of measurement error. We propose an analytical framework and methodology that is robust to misspecification of the first-stage model and provides valid inference for the second-stage model parameter of interest. 

We decompose the measurement error into components analogous to classical and Berkson error and characterize properties of the estimator in the second-stage model if the first-stage model predictions are plugged in without correction. Specifically, we derive conditions for compatibility between the first- and second-stage models that guarantee consistency (and have direct and important real-world design implications), and we derive an asymptotic estimate of finite-sample bias when the compatibility conditions are satisfied. We propose a methodology that (1) corrects for finite-sample bias and (2) correctly estimates standard errors. We demonstrate the utility of our methodology in simulations and an example from air pollution epidemiology.

\end{abstract}

\section{Introduction}
\label{se:intro}

We consider measurement error that results from using predictions from a first-stage statistical model as the covariate of interest (the exposure) in a second-stage association study. Regardless of the exposure prediction model, there will be measurement error from the difference between predictions and the unmeasured true values. In contrast with standard measurement error models and the usual classification into classical and Berkson error, such predictions induce a complicated form of measurement error that is heteroscedastic and correlated across study subjects \citep{Gryparis2009,Szpiro2011biostats}. Our objectives are to
characterize the effects of this error, to give guidelines for study design to minimize the impact, and to provide a correction method that reduces bias and gives valid confidence intervals.  

In this section, we review the literature on measurement error correction for air pollution cohort studies and describe how our approach advances the state of the art (Section~\ref{se:intro1}), comment on connections between our work and fundamental statistical issues concerning  
the interpretation of random effects models and the interplay between random vs. fixed covariate  
regression and misspecified mean models (Section~\ref{se:intro2}), and outline the main sections of the paper (Section~\ref{se:intro3}).

\subsection{Measurement error} 
\label{se:intro1}
There has been extensive research on measurement error \citep{Carroll2006}, but the statistical literature has only recently begun to deal with the problem presented here. For spatial exposure contrasts, we have generalized the standard categories by decomposing the measurement error into a Berkson-like component from smoothing the exposure surface and a classical-like component from variability in estimating exposure model parameters  \citep{Gryparis2009,Szpiro2011biostats,Sheppard2011}. 
We and others have also shown that the parametric bootstrap (or a computationally efficient approximation to the parametric bootstrap) 
can be used to correct for the effects of measurement error  \citep{Madsen2008, lopiano2011comparison,Szpiro2011biostats}.   However, validity 
of these results depends crucially on having a correctly specified exposure model. In practice such models are developed for predictive performance and often use predictors based on convenience, so we believe misspecification is ubiquitous. 
A distinguishing feature of our methodology in this paper is that it is robust to misspecification of the mean and/or variance in the exposure model and still provides valid second-stage inference. 

We focus on two-stage analysis, as this is a common and practical approach when exposure is not directly measured. An alternative is a unified analysis in which the exposure model is a component of a joint model for the exposure and health data (e.g., \cite{Sinha2010} in nutritional epidemiology and \cite{Gryparis2009} in air pollution epidemiology), but this type of joint model has several difficulties. First, it presupposes that one has a correct (or at least nearly correct) exposure model; we argue that an exposure model can generally only capture a portion of the full exposure and should be treated in this light. Second, outlying second-stage data may influence estimation of the exposure model in unexpected ways, especially when the second-stage model is misspecified (noting at the same time that this feedback is an essential aspect of a coherent joint model and leads to increased efficiency). Third, the same exposure data are often used with multiple second-stage outcome data, and it is scientifically desirable to use the same predicted exposures across studies.  Finally, exposure modeling can be computationally demanding, involving spatial and spatio-temporal prediction, so pursuing a two-stage strategy has practical appeal.
For further elaboration of these points, see \citet{bennett2001errors}, \citet{wakefield2006health}, \citet{Gryparis2009}, \citet{lunn2009combining}, and \citet{Szpiro2011biostats}.

We and others have also evaluated standard correction methods such as regression calibration, including using personal measurements as validation data \citep{Gryparis2009, Spiegelman2010, Szpiro2011biostats}. Performance in the spatial setting is mixed, most likely because the error structure differs substantially from classical measurement error. 
The methodology we describe here directly accounts for the spatial characteristics of measurement error and relies on statistical estimates of uncertainty from the exposure model rather than validation data.

A key application, and the one that motivates this work, is studying the health effects of chronic exposure to ambient air pollution.  Long-term air pollution exposure has been linked with
increased cardiovascular morbidity and mortality in prominent studies that form part of the basis for regulations with broad economic impact \citep{Dockery1993,Pope2002,Peters2002}.  
Early air pollution cohort studies focused on mortality \citep{Dockery1993,Pope2002}, while
more recent work has shown associations with non-fatal cardiovascular events \citep{Miller2007} and sub-clinical indicators of
disease \citep{Kunzli2005,VanHee2009,Adar2010,van2011association}. In general, the health risk of air pollution to any single individual is thought to be small, but there are important public health implications because of the large number of people exposed and the ability of governments to mitigate exposure
through regulatory action \citep{Pope2006}. 

In air pollution studies, exposure modeling is motivated by the desire to estimate intra-urban (i.e., within a metropolitan area) variation in exposure, which is 
more difficult to quantify than inter-urban pollution contrasts. There are significant advantages to exploiting intra-urban contrasts, as this can increase statistical power to detect health effects, help rule out unmeasured confounding by city or region, and improve our ability to differentiate between the effects of different pollutants or pollutant components. Pollution data are typically available from regulatory and research monitoring networks but not from long-term residential or personal monitoring of individuals participating in observational health studies, leading to a spatial misalignment problem. Typical exposure prediction models rely on monitoring data in a regression with geographically-varying covariates and smoothing by splines or kriging \citep{Fanshawe2008,Jerrett2005a, Hoek2008, Su2009, Yanosky2009, Szpiro2009, Brauer2010}. Standard practice is to select an exposure model with good prediction accuracy, treat the predicted exposures as known, and plug them into a health model to estimate the association of interest without accounting for measurement error \citep{Jerrett2005, Kunzli2005,kim2009,Puett2009,Adar2010,van2011association}.

While motivated by air pollution epidemiology, the core 
measurement error 
ideas in this paper have much broader relevance. Indeed, \cite{Prentice2010} (in the 2008 RA Fisher lecture at the Joint Statistical Meetings) states that ``measurement error in exposure assessment may be a potentially dominating source of bias in such important prevention research areas as nutrition and physical activity epidemiology.''  It is essential to better understand the implications of measurement error in a wide variety of applications in which one must first estimate exposure. These applications include (1) nutritional epidemiology, (2) physical activity epidemiology, (3) a environmental and occupational epidemiology, (4) exposure to disease vectors or infectious agents, and (5) two-stage analyses in functional data contexts.  Statistical exposure models are commonly used in environmental and occupational epidemiology \citep{Dement1983,Preller1995,Stram1999, Ryan2007, Slama2007}, with kriging and land use regression particularly popular in air pollution research. 
More generally, proxy data are becoming increasingly available and a natural idea in many contexts is to model an exposure of interest given publicly available data. Such data could include remote sensing from satellites or large networks of inexpensive sensors deployed to measure physical phenomena. 

\subsection{Connections to other fundamental statistical issues}
\label{se:intro2}
In addition to advancing measurement error research, our development emphasizes the relationships between certain foundational issues in applied statistics that are of current interest in the field, specifically 
the interpretation of random effects models and the interplay between random vs. fixed covariate  
regression and misspecified mean models.  

As discussed in Section~\ref{se:dgm}, we have chosen to condition on a fixed
but unknown spatial air pollution surface, rather than taking the more conventional geostatistical
approach of modeling
an unknown spatial surface as a random effect or spatial random field \citep{Cressie1993, Banerjee2004}. 
The repercussions of this decision are related to the more general question of how to interpret random effects models in light of reasonable assumptions about the true data-generating mechanism, and whether this terminology is even adequate for describing the range of problems to which random effects-based algorithms are currently applied \citep{Gelman2005, Hodges2010}.  Indeed, in 
a new book on richly parameterized models, \citet{Hodges2013} points out that our 
particular 
modeling framework
illustrates an important practical difference in inferential methodology between what he calls `old' and `new' style random effects.  

As discussed in Sections~\ref{se:dgm}--\ref{se:exposEst}, we regard the entire unknown exposure surface as part of the mean
in a finite rank regression, rather than allocating the spatial component to the variance by means of a random effect, so we must
address the consequences of a misspecified mean model.
In addition, we regard the exposure monitor locations as random rather than fixed (since they could 
presumably vary between hypothetical repeated experiments), so we are in the
setting of random covariate regression with a misspecified mean model.
\citet{Buja2013} and \citet{Szpiro2010a} have recently discussed some implications of the distinction
between fixed and random covariates when the mean model is misspecified.  
In fact, the seminal paper by \citet{White1980} on sandwich covariance estimators includes the case of a misspecified mean model, but perhaps in part because the title focuses on heteroscedasticity, applications of the sandwich estimator tend to focus only on the importance of non-constant variance.
One often neglected consequence of the ``conspiracy of
model violation and random X'' \citep{Buja2013} that is important in our development is that regression parameter
estimates are not quite unbiased.  We provide an approach to characterizing and estimating the asymptotic bias 
(see equation (\ref{eq:biasgammahat.body}) and the surrounding discussion) that is, as far as we know, novel.

\subsection{Outline of paper}
\label{se:intro3}
Section 2 presents our basic framework, a key feature of which is that it avoids the assumption that the exposure model is correct and instead projects exposure data 
into a lower dimensional space.  We present conditions on the compatibility of the first and second stage designs that have important real-world design and analysis implications.
Section 3 decomposes the resulting measurement error into Berkson-like and classical-like components. Under the compatibility conditions, we show that there is 
essentially 
no bias from the Berkson-like error, although this component of the error still increases variability of second-stage effect estimates.  We then derive asymptotic estimates of the bias and variance caused by the classical-like error. 
Section 4 describes our measurement error correction approach, wherein we correct for bias from the classical-like error using our asymptotic results and estimate the uncertainty, including that from both sources of measurement error, using 
a form of the nonparametric bootstrap. 
Sections 5 and 6 present simulations and an example application to the Multi-Ethnic Study of Atherosclerosis and Air Pollution (MESA Air).

\section{Analytical framework}

\subsection{Data-generating mechanism}
\label{se:dgm}

We develop an analytic framework for air pollution cohort study data
that also more generally illustrates how one can use a measurement
error paradigm to formalize two-stage analysis with a misspecified
first-stage model.  While previous work has modeled the spatial
variation in air pollution as a random field
\citep{Gryparis2009,Szpiro2009,Szpiro2011biostats}, we regard the
spatial surface itself as fixed and treat the data locations as
stochastic.  This avoids the philosophical difficulties inherent in
attributing spatial structure of long-term average air pollution to a
stochastic spatial process that would be different in a hypothetical
repeated experiment.  
Long-term average air pollution concentrations
over one or more years
are predominately determined by fixed but complex climatological,
economic, and geographic systems, so it is scientifically preferable
to regard the unknown surface as deterministic.
Thus, we condition on the fixed physical world in the time period of the study and consider a repeated
sampling framework in which observations might have been collected at
different locations according to a (not necessarily known) study
design. In Section~\ref{se:disc}, we discuss the implications of this approach when considering shorter-term air pollution exposures. 

More formally, consider an association study with health outcomes $y_i$ and corresponding exposures $x_i$ for subjects $i=1,\ldots,n$ at geographic  locations
$\s_i \in \R^2$,  with additional health model covariates $\z_i =(z_{i1},\ldots,z_{ip}) \in \R^p$, 
including an intercept. Consider a linear model,
\be
\label{eq:healthlin}
y_i=x_i\beta+\z_i \bbeta_z +\epsilon_i,
\ee
where conditional on covariates the $\epsilon_i$ are independent but not necessarily identically distributed,
satisfying $E(\epsilon_i)=0$.
Our target of inference is the health effect parameter, $\beta$.
If the $x_i$ and $\z_i$  were observed without error, inference for $\beta$ would be routine by ordinary least squares (OLS)
and sandwich-based standard error estimates \citep{White1980}.
We are interested in the situation where the $y_i$ and $\z_i$ are observed for all subjects, but instead of the actual subject exposures we observe 
monitoring data, $x^*_j$, for $j=1,\ldots,n^*$, at different locations $\s^*_j$. Nonlinear health models are of course important and are the subject of ongoing research, but the linear setting is helpful for developing the general framework and our specific asymptotic results.

We emphasize that we regard the spatial locations  $\s_i$ and $\s^*_j$ of study subjects  and monitors as realizations of spatial random variables.  
The locations are chosen at the time of the study design, and it is natural to regard them as stochastic in order to address the statistical question of how the
estimates of $\beta$ would vary if different locations were selected according to similar criteria. 
Thus, in our development we assume the $\s_i$ and $\s^*_j$  
are distributed in $\R^2$ with unknown densities $g(\s)$ and  $h(\s)$, respectively, and corresponding distribution functions $G(\s)$ and $H(\s)$.
Throughout, we assume the subject locations are chosen independently of the monitoring locations. To simplify the exposition, we further assume in Sections~\ref{se:decomp} and~\ref{se:correction}
that both sets of locations are i.i.d.  It is straightforward to account for clustering of subject or monitor locations; see, for example, the simulation study in Section~\ref{se:spatial.ex} 
and the data analysis in Section~\ref{se:mesa}.

Conditional on the $\s_i$, we assume the $x_i$ satisfy
\bes
x_i = \Phi(\s_i) + \eta_i,
\ees
with i.i.d. mean zero $\eta_i$.  The function $\Phi(\s)$ is a deterministic spatial surface that is potentially predictable
by covariates and spatial smoothing, and the $\eta_i$  represent 
variability between exposures for subjects at the same physical location. 
We assume an analogous model for the monitoring data at locations $\s^*_j$, with the same deterministic spatial field $\Phi(\s_j^*)$ and with instrument
error represented by $\eta_j^*$ having variance $\sigma^2_{\eta^*}$. 

Finally, we assume the additional health model covariates $\z_i$ satisfy
\bes
\z_i = \bTheta(\s_i) + \bzeta_i,
\ees
where  $\bTheta(\s)=(\theta_1(\s),\ldots,\theta_p(\s))$ is a $p$-dimensional vector-valued function representing the spatial
component of the additional covariates, and which includes the intercept, and the $\bzeta_i=(\zeta_{i1},\ldots,\zeta_{ip})$ are random $p$-vectors 
independent between subjects and independent of the $\eta_i$.  Each component of $\bzeta_i$ has mean zero, but the components
of $\bzeta_i$ are not necessarily independent of each other.
To illustrate, one additional health model covariate might be household income, decomposed into spatial variation representing the socioeconomic status of the neighborhood and the residual variation between residences.

\subsection{Exposure estimation}
\label{se:exposEst}
Standard practice is to derive a spatial estimator of exposure $\hat{w}(\s)$ based on the monitoring data and then
to use the $\hat{w}(\s_i)$ in place of the $x_i$ in~(\ref{eq:healthlin}) to estimate $\beta$. 
We consider a hybrid regression (on geographically-defined covariates) and regression spline exposure model. 
Thus, we let $\capr(\s)$ be a known function from $\R^2$ to $\R^r$ that incorporates $q$ covariates and $r-q$
spline basis functions.  If we knew the least-squares fit of the exposure surface with respect to the density of subject locations $g(\s)$,
\be
\label{eq:gamma}
\bgamma = \argmin_\bxi \int \big(\Phi(\s)-\capr(\s)\bxi\big)^2 dG(\s),
\ee
it would be natural to approximate $x_i$ by $w(\s_i)=\capr(\s_i)\bgamma$.
Notice that we do not assume the spatial basis
is sufficiently rich to represent all of the structure in $\Phi(\s)$, so we allow for misspecification in the sense that $\Phi(\s) \neq  w(\s)$ for some $\s \in \R^2$,
for any choice of $\bgamma$. 

We do not know $\bgamma$, so we will estimate it from the monitoring data by $\hat{\bgamma}$ and then use the estimated exposure, $\hat{w}(\s_i)=\capr(\s_i)\hat{\bgamma}$,
in place of  $x_i$.  In particular, we derive $\hat{\bgamma}$ by OLS
\be
\label{eq:gammahat2}
\hat{\bgamma} = \argmin_\bxi \sum_{j=1}^{n^*} \left(x_i^*-\capr(\s^*_j)\bxi\right)^2.
\ee
Under standard regularity conditions \citep{White1980}, $\hat{\bgamma}$ is asymptotically normal and converges a.s. to $\bgamma^*$ as $n^*\goto\infty$,
where $\bgamma^*$ is the solution to (\ref{eq:gamma}) with $H(\s)$ in place of $G(\s)$.
In Section~\ref{se:compatibility}, we discuss the implications of distinct reference distributions in (\ref{eq:gamma}) and (\ref{eq:gammahat2}).

\subsection{Exposure model choice}
\label{se:compatibility}
So far we have taken $\capr(\s)$ to be a known function from $\R^2$ to $\R^r$,
encoding a set of decisions about which covariates and spline basis functions 
to include in the exposure model. 
Indeed, model selection is a complex task that involves trading off 
flexibility 
(advantageous for modeling as much of the true exposure surface as possible) and
parsimony (advantageous for reducing estimation error).  
We begin by specifying compatibility conditions for the first-stage exposure model that are needed to guarantee consistent estimation of $\beta$ in the second-stage health model.  The following two conditions are sufficient,
and we will discuss their motivation further in Section~\ref{se:decomp}.
\begin{cond}
\label{cond1}
The probability distribution of $\capr(\s)$ is the same if $\s$ is sampled from $G(\s)$ or $H(\s)$.
\end{cond}
\begin{cond}
\label{cond2}
The span of $\capr(\s)$ includes the elements of $\bTheta(\s)$, $\theta_k(\s), k=1,\ldots,p$, the  spatially structured components of the additional health model covariates. 
\end{cond}
Note that Condition~\ref{cond1} is satisfied if 
the probability distributions of subject and monitor locations are identical, i.e., $g(\s)=h(\s)$ for all $\s$.
Visual inspection on a map can be useful for verifying that $g(\s)$ and $h(\s)$ 
represent similar spatial patterns that are relevant for spline functions, but individual geographic covariates
may have very fine spatial structure, so it is also useful to examine the values of these geographic covariates at subject and monitor locations.  If a particular covariate has noticeably different distributions
in the two populations, then it should not be included in $\capr(\s)$ (see, for example, the discussion of the MESA data analysis in Section~\ref{se:mesa}). 

Selecting $\capr(\s)$ to satisfy Condition~\ref{cond2} implicitly requires that $\bTheta(\s)$ be defined at all locations in the supports of $g(\s)$ and $h(\s)$.
If $g(\s)=h(\s)$ for all $\s$, then this is automatically true since $\bTheta(\s)$ is defined at all locations where it is possible
for study subjects to be located.

Beyond the compatibility conditions above, there is a sizable and relevant statistical literature
on methods for maximizing out-of-sample prediction accuracy, which for spline models 
amounts to selecting the number of basis functions and locations of knots 
or selecting a penalty parameter 
\citep{Hastie2001, ruppert2003semiparametric}.
In our setting, improved accuracy of exposure model predictions will often correspond to improved
efficiency in estimating $\beta$, although this is not always the case \citep{Szpiro2011epi}.
We comment further on the tradeoff between exposure model complexity and parsimony
in Section~\ref{se:disc}, but a specific algorithm for selecting geographic 
covariates or spline basis functions is beyond the scope of this paper.


\section{Measurement error}
\label{se:decomp}

Let $\hat{\beta}_{n,n^*}$ be the health effect estimate obtained from the OLS solution to (\ref{eq:healthlin}) using $\hat{w}(\s_i)$ estimated from $n^*$ monitoring locations in place of $x_i$, for study subjects $i=1,\ldots,n$. This estimator is affected by two fundamentally different types of measurement error: Berkson-like and classical-like components \citep{Szpiro2011biostats}. Defining $w^*(\s_i)=\capr(\s_i)\bgamma^*$, we can express the measurement error, $u_{i}=x_i-\hat{w}(\s_i)$, as
\bea
u_{i}&=&\big(x_i- w^*(\s_i)\big) + \big(w^*(\s_i)-\hat{w}(\s_i) \big)  \nonumber\\
&=&u_{i,BL}+u_{i,CL} \label{eq:decomp}.
\eea
The Berkson-like component, $u_{i,BL}$, is the information lost from smoothing even with unlimited monitoring data (a form of exposure model misspecification), and the 
classical-like component, $u_{i,CL}$, is variability that arises from estimating the parameters of the exposure model based on monitoring data at $n^*$ locations.

The designation of $u_{i,BL}$ as Berkson-like error refers to the fact that
this is part of the true  exposure surface that our model
is unable to predict, even in an idealized situation with unlimited monitoring data.
As such, it results in predictions that are less variable than truth. In Section~\ref{se:uBL} we consider the impact of the Berkson-like error alone and demonstrate asymptotic unbiasedness
for large $n$ in Lemma~\ref{le:ubl},
assuming the compatibility conditions of Section~\ref{se:compatibility} are satisfied. This result motivates the need for the compatibility conditions, but it is not used directly in our measurement error methodology
in Section~\ref{se:correction}. 
Our consistency result in Lemma~\ref{le:ubl} is analogous to Lemma~1 in \citet{White1980}, indicating that finite sample bias occurs in generic random covariate regression even in the absence of measurement error. Here we regard this bias as negligible, because in public health contexts $n$ is often relatively large, particularly compared to $n^*$. Although $u_{i,BL}$ alone does not induce important bias, it does inflate the variability of health
effect estimates, and we account for this with the nonparametric bootstrap in our proposed measurement error methodology in Section~\ref{se:correction}. 


Classical-like measurement
error, $u_{i,CL}$, results from the finite $n^*$ variability of $\hat{\bgamma}$
as an estimator of $\bgamma^*$. As discussed by \citet{Szpiro2011biostats}, it is similar to classical measurement error in the sense that it contributes
additional variability to exposure estimates that is not related to the outcome.  Like classical measurement error,
$u_{i,CL}$ introduces bias in estimating $\beta$ and affects the standard error, but it is not the same as classical measurement error because it is heteroscedastic and shared between subjects. In Section~\ref{se:uCL} we estimate the bias from classical-like measurement error 
(still under the  the compatibility conditions of Section~\ref{se:compatibility}). 
This estimate will provide a means 
to correct for bias as part of our measurement error methodology in Section~\ref{se:correction}.

\subsection{Berkson-like error ($u_{i,BL}$)}
\label{se:uBL}

Considering our estimator, $\hat{\beta}_{n,n^*}$, we isolate the impact of $u_{i,BL}$ by operating in the $n^*= \infty$ limit 
with $w^*(\s_i)=\capr(\s_i)\bgamma^*$ and analyzing 
the behavior of $\hat{\beta}_{n,\infty}$. 
The following lemma holds under sufficient regularity of $g(\s)$, $h(\s)$, $\Phi(\s)$, and $\capr(\s)$.
We include the proof here because it is helpful for understanding the importance of the compatibility conditions in Section~\ref{se:compatibility}.

\begin{lem}
\label{le:ubl}
Assuming Conditions~\ref{cond1} and~\ref{cond2}, $\hat{\beta}_{n,\infty}$
converges a.s. to $\beta$ as $n\goto \infty$.
\end{lem}
\begin{proof}
It is easy to see that $\hat{\beta}_{n,\infty}$ is the OLS solution to (\ref{eq:healthlin}) using $w^*(\s_i)=\capr(\s_i)\bgamma^*$ in place of $x_i$.
Condition~\ref{cond1} implies $\bgamma^*=\bgamma$, so we consider the impact of using $w(\s_i)=\capr(\s_i) \bgamma$ as the exposure.  We write
\be
\label{eq:ubl}
y_i=\capr(\s_i) \bgamma \beta + \z_i \bbeta_z + \left(\left(\Phi(\s_i)-\capr(\s_i) \bgamma\right)\beta + \eta_i \beta+  \ve_i\right),
\ee
where the three terms grouped in parentheses are regarded as unobserved error terms.
To apply Lemma~1 from \citet{White1980}, it is sufficient that
\be
E\left\{\capr(\s_i) \bgamma  \times (\Phi(\s_i)-\capr(\s_i) \bgamma)\right\}=0  \label{eq:exporthog}
\ee
and for each $k=1,\ldots,p$
\be
E\left\{\theta_k(\s_i) \times (\Phi(\s_i)-\capr(\s_i) \bgamma)\right\}=0, \label{eq:zorthog}
\ee
where the random sampling of $\s_i$ is according to the density of subject locations, $g(\s)$.   Orthogonality of residuals in the least squares optimization for $\bgamma$ in (\ref{eq:gamma}) implies  (\ref{eq:exporthog}), and Condition~\ref{cond2} implies (\ref{eq:zorthog}) since each $\theta_k(\s)$ can be represented as a linear combination of elements of $\capr(\s)$.  We actually need (\ref{eq:zorthog}) with $z_{ik}=\theta_k(\s_i) + \zeta_{ik}$ in place of $\theta_k(\s_i)$ for Lemma~1 of \citet{White1980}, but this follows 
from (\ref{eq:zorthog}) since $\zeta_{ki}$ has mean zero and is independent of $\s_i$.  
\end{proof}

We comment on the necessity of Conditions~\ref{cond1} and~\ref{cond2}.
The proof of Lemma~\ref{le:ubl} depends on  $\bgamma^* = \bgamma$.  This will always hold if $\Phi(\s)$ is spanned by the $\capr(\s)$, but otherwise we rely on Condition~\ref{cond1}.
If $\bgamma^* \neq \bgamma$,  then (\ref{eq:ubl}) becomes
\be
\label{eq:ubl.mismatch}
y_i=\capr(\s_i) \bgamma^* \beta + \z_i \bbeta_z + \left( \left(\Phi(\s_i)-\capr(\s_i) \bgamma^*\right)\beta + \eta_i \beta+  \ve_i\right).
\ee
We cannot expect that  $\capr(\s_i) \bgamma^*$ is orthogonal to $\left(\Phi(\s_i)-\capr(\s_i) \bgamma^*\right)$
when $\s_i$ is drawn according to the probability density $g(\s)$, since $\gamma^*$ is the least squares fit from (\ref{eq:gamma}) with $H(\s)$
in place of $G(\s)$.  Therefore, treating $\left(\Phi(\s_i)-\capr(\s_i) \bgamma^*\right)$ as part of the random
variation in (\ref{eq:ubl.mismatch}) results in the equivalent of omitted variable bias when estimating $\beta$.

Condition~\ref{cond2} is needed to guarantee (\ref{eq:zorthog}) in the proof of Lemma~\ref{le:ubl}.
The difficulty if (\ref{eq:zorthog}) does not hold is that
$\left(\Phi(\s_i)-\capr(\s_i) \bgamma\right)$ in (\ref{eq:ubl})
may be correlated with one or more elements of the $\bTheta(\s_i)$ component of $\z_i$.  
Intuitively, this can introduce bias because estimation of $\beta$ relies on the variation in $\capr(\s_i)\bgamma$ that is unrelated to the covariates $\z_i$, i.e., the residual variation after projecting onto the span of the elements of $\z_i$, which is equivalent to the span of $\bTheta(\s_i)$.  Without  (\ref{eq:zorthog}), the residual term $\left(\Phi(\s_i)-\capr(\s_i) \bgamma\right)$ in
(\ref{eq:ubl}) need not be orthogonal to this variation.  
Note that the need to include the  covariates from the health model in the exposure model is analogous to the inclusion of covariates in standard regression calibration \citep{Carroll2006}.

\subsection{Classical-like error ($u_{i,CL}$)}
\label{se:uCL}

We will isolate the impact of $u_{i,CL}$ on $\hat{\beta}_{n,n^*}$ by operating in the $n= \infty$ limit, 
corresponding to the entire superpopulation of study subjects, and analyzing 
the asymptotic properties of $\hat{\beta}_{\infty,n^*}$ as $n^* \rightarrow \infty$.  
The exposure model parameter vector, $\hat{\bgamma}$, is asymptotically normal (as discussed in Section~\ref{se:exposEst}) with dimension fixed at $r$, and  $\hat{\beta}_{\infty,n^*}$ is a deterministic function of $\hat{\bgamma}$, 
so under the conditions of Lemma~\ref{le:ubl} a standard delta-method argument can be used to establish that $\hat{\beta}_{\infty,n^*}$ is asymptotically normal with mean
 $\beta$.  
In particular, this implies that bias from classical-like error is asymptotically negligible in the sense that it is of comparable magnitude to the
variance.  This situation contrasts with classical measurement error where 
there are as many random error terms
as observations and there is large-sample bias \citep{Carroll2006}.

Even though the bias term is asymptotically negligible, our simulation studies suggest that it can still be important for moderate size $n^*$, so we will derive a bias correction.  
Since only the variability in the exposure estimate that is orthogonal to covariates from the health model
plays a role in deriving $\hat{\beta}_{\infty,n^*}$, it is helpful in the following analysis to
define $\capr^c(\s)$ with elements  $R_k^c(\s)=R_k(s)-\bTheta(\s)\bpsi_k$, where $\bpsi_k = \argmin_\bomega \int (R_k(\s)-\bTheta(\s)\bomega\big)^2 dG(\s)$. 
Analogous to $\hat{w}(\s)$ and $w(\s)$, we define $\hat{w}^c(\s)=\capr^c(\s)\hat{\bgamma}$ and $w^{c}(\s)=\capr^c(\s)\bgamma$.

Note that the expectation of 
$\hat{\beta}_{\infty,n^*}$  need not be defined for finite $n^*$ because it is a function of  $\hat{\bgamma}$ , and the denominator in the OLS solution for $\hat{\bgamma}$ is not bounded away from zero.
Therefore, we adapt the definition of asymptotic expectation for a sequence of random variables from \citet[page 135]{Shao2010}.  The basic idea is to
identify the highest order term in a power series expansion that has non-zero expectation as the asymptotic expectation.  See a related discussion of concepts of asymptotic bias in \citet[Appendix A.1.2]{Lumley2010}.  

\begin{defin}
\label{def:asympt}
Let $\upsilon_1,\upsilon_2,\ldots$ be a sequence of vector-valued random variables and let $a_1,a_2,\ldots$ be a sequence of positive numbers such that $\lim_{n\goto\infty}a_n = \infty$.
(i) Suppose $\upsilon$ is such that $E|\upsilon|<\infty$ and we can write $\upsilon_n =\tilde{\upsilon}_n+ \upsilon^\prime_n$ with $E(\tilde{\upsilon}_n)$=0 and $\lim_{n\goto \infty} a_n \upsilon^\prime_n \goto_d \upsilon$.  Then we denote
$E_{[a_n]}(\upsilon_n)= E(\upsilon)$ and call $E_{[a_n]}(\upsilon_n)/a_n$ an order $a_n^{-1}$ asymptotic expectation of $\upsilon_n$.
(ii)  Suppose $\upsilon$ is such that $\C(\upsilon)<\infty$ and $\lim_{n\goto \infty} \sqrt{a_n} \upsilon_n \goto_d \upsilon$.  Then we denote
$\C_{[a_n]}(\upsilon_n)= \C(\upsilon)$ and call  $\C_{[a_n]}(\upsilon_n)/a_n$ an order $a_n^{-1}$ asymptotic covariance of $\upsilon_n$.
\end{defin}

\begin{lem}
\label{le:ucl}
Assume sufficient regularity of $g(\s)$, $h(\s)$, $\Phi(\s)$, and $\capr(\s)$ and Conditions~\ref{cond1} and~\ref{cond2}.   If we set
\bea
E_{[n^*]}(\hat{\beta}_{\infty,n^*} - \beta)
&=&\beta \Big \{ -\frac{\int w^{c}(\s) E_{[n^*]} \left(\hat{w}^c(\s)-w^{c}(\s)\right) dG(\s)}{\int w^{c}(\s)^2 dG(\s)} -\frac{\int \V_{[n^*]}\left(\hat{w}^c(\s)\right) dG(\s)}{\int w^{c}(\s)^2 dG(\s)} + \nonumber \\
\label{eq:cl.bias}\\
&&  
2 \frac{\int w^{c}(\s_1) w^{c}(\s_2) \C_{[n^*]} \left(\hat{w}^c(\s_1), \hat{w}^c(\s_2)\right) dG(\s_1) dG(\s_2)}{\left(\int w^{c}(\s)^2 dG(\s)\right)^2} \Big \}\nonumber
\eea
\newline
\noindent 
and  \nonumber
\bea
\V_{[n^*]}(\hat{\beta}_{\infty,n^*})&=&\beta^2 \Big\{\frac{\int w^{c}(\s_1) w^{c}(\s_2) \C_{[n^*]} \left(\hat{w}^c(\s_1), \hat{w}^c(\s_2)\right) dG(\s_1) dG(\s_2)}{\left(\int w^{c}(\s)^2 dG(\s)\right)^2}  \Big \},  \label{eq:cl.var}
\eea
then $E_{[n^*]}(\hat{\beta}_{\infty,n^*} - \beta)/n^*$  is an asymptotic expectation of $\hat{\beta}_{\infty,n^*} - \beta$ and $\V_{[n^*]}(\hat{\beta}_{\infty,n^*})/n^*$ is an asymptotic variance of $\hat{\beta}_{\infty,n^*}$ (both of order  ${n^*}^{-1}$) .
\end{lem}
The proof is outlined in Appendix~\ref{se:appendixa}, where we express $\hat{\beta}_{\infty,n^*}$ as a function of $\hat{\bgamma}$ and do a second order Taylor expansion around $\bgamma$. 
Definition~\ref{def:asympt} is required to define the order  ${n^*}^{-1}$ asymptotic expectation in the first term of (\ref{eq:cl.bias}), which is a linear function of $\hat{\bgamma}-\bgamma$. The first order terms in a Taylor expansion of $\hat{\bgamma}-\bgamma$ are of order $n^{*-1/2}$ and do not converge when multiplied by $n^*$. However, they have expectation zero, so they play the role of $\tilde{\upsilon}_n$ and do not contribute to the asymptotic expectation.  See (\ref{eq:biasgammahat.body}) and the surrounding discussion. 

The practical import of Lemma~\ref{le:ucl} is that we can use (\ref{eq:cl.bias}) to correct for the bias from classical-like error. The variance estimate in (\ref{eq:cl.var}) is not directly useful as a standard error because it does not include variability from Berkson-like error or from having $n<\infty$ study subjects, but it provides insight into the relative magnitudes of bias and variance from classical-like error.

To estimate (\ref{eq:cl.bias}), we can estimate $w^{c}(\s)$ by $\hat{w}^c(\s)=\capr^c(\s)\hat{\bgamma}$, noting that $\capr^c(\s)$ is approximated from the observed
exposure covariates for the health observations, orthogonalizing with respect to the health model covariates, which is the finite sample approximation to the construction of $\capr^c(\s)$ stated earlier in this section.
To estimate the variances and covariances of $\hat{w}^c(\s)$, we use a  robust estimator for $\mbox{Cov}_{[n^*]}(\hat{\bgamma})$ \citep{White1980, Carroll2006}.
We use the sandwich estimator to avoid the assumption of having a correctly-specified model, as required for the standard model-based estimator. Given these estimators, all the integrals in the first two terms of (\ref{eq:cl.bias}) can be estimated as averages with respect to the discrete measure with equal weight on each health observation, the standard plug-in estimator for $G(\s)$. 

Finally, in the third term of (\ref{eq:cl.bias}), we need to estimate $E_{[n^*]}(\hat{w}^c(\s) - w^{c}(\s)) = \capr(\s)E_{[n^*]}(\hat{\bgamma} - \bgamma)$, and therefore the asymptotic expectation of $\hat{\bgamma}$.  
Since we have assumed Condition~\ref{cond1}, which implies $\bgamma=\bgamma^*$, the expectation of $\hat{\bgamma}$ is approximately equal to $\bgamma$.  However, $\hat{\bgamma}$
is derived by means of a random covariate regression with a misspecified mean model, so its standard expectation is not defined.  An estimate of its asymptotic expectation is developed as follows.
Let  $\bPhi^*$ be the vector comprised of the $\Phi(\s^*_j)$ and $\capr^*$ the $n^* \times r$ matrix obtained
by stacking the $\capr(\s^*_j)$ for $j=1,\ldots,n^*$.  
For arbitrary $m_j$, denote by $\capm$ the $n^*\times n^*$ diagonal matrix with entries
$m_1,\ldots,m_{n^*}$.
If we set $m_j=1/n^*$ for $j=1,\ldots,n^*$ and define
\be
\kappa(m_1,\ldots,m_{n^*}) = \big(\capr^{*^ \top} \capm \capr^*\big)^{-1}\capr^{*^\top} \capm\bPhi^* \label{eq:kappa},
\ee
then we notice
$E\left(\hat{\bgamma}|\s^*_1,\ldots,\s^*_{n^*}\right)=\kappa(m_1,\ldots,m_{n^*})$.
We are interested in the unconditional expectation of $\hat{\bgamma}$.
Heuristically, we assume that the true $h(\s)$ is supported on the observed monitor locations and gives equal weight to each observation (i.e., we use the plug-in estimator for $h(\s)$).
In that case, a realization of $\s^*_1,\ldots,\s^*_{n^*}$ can be expressed as a multinomial draw, $m_1,\ldots,m_{n^*}$, where the $m_j$ are the fraction of times each location in the support of $h(\s)$
is drawn.  We can estimate
the expectation of $\kappa(m_1,\ldots,m_{n^*})$ by means of a Taylor series expansion of $\kappa(\cdot)$ around $m_j=\frac{1}{n^*}$ for $j=1,\ldots,n^*$. Using the first and second moments of a multinomial distribution, we have
\be
\label{eq:biasgammahat.body}
E_{[n^*]} \left(\hat{\bgamma} - \bgamma\right) \approx \frac{1}{2}\left(\frac{1}{n^*}-\frac{1}{(n^*)^2}\right)\sum_{j=1}^{n^*} \frac{\partial^2 \kappa}{\partial m_j^2} -
\frac{1}{2} \frac{1}{(n^*)^2} \sum_{j,k=1;j\neq k}^{n^*} \frac{\partial^2 \kappa}{\partial m_j \partial m_k}.
\ee
It easy to see that the first order terms in the Taylor expansion of $\kappa(\cdot)$ (not shown) have expectation zero, so they play the role of $\tilde{\upsilon}_n$ in Definition~\ref{def:asympt}
and do not contribute to the asymptotic expectation.
We give further details on numerical calculation of the above expression in Appendix~\ref{se:appendixb}.  
A more formal derivation that does not begin by assuming a discrete distribution could be developed by a von Mises expansion with the empirical process of monitor locations \citep[Section 20.1]{Vaart1998}.
Note that although we do not observe the $\Phi(\s^*_j)$, replacing them with $x^*_j$ in (\ref{eq:kappa}) does not introduce bias since $x^*_j=\Phi(\s^*_j)+\eta^*_j$, and the $\eta^*_j$ are 
independent of everything else and have mean zero.

Finally, we can gain additional insight into the bias and variance contributions from classical-like error by considering the simplified situation in which the exposure model is correctly specified so that $w(\s)=\capr(\s) \bgamma$ for all $\s$,
the subject and monitor location densities, $g(\s)$ and $h(\s)$, are the same, and there are no additional covariates or intercept in the health model.  In that case
it is easy to show that the asymptotic expectation simplifies to 
\bea
\label{eq:simp.bias}
\frac{1}{n^*}E_{[n^*]}(\hat{\beta}_{\infty,n^*} - \beta)
&=&- \beta \frac{1}{n^*} \frac{(r-2)\sigma^2_{\eta^*}}{\int w(\s)^2 dG(\s)},
\eea
and the asymptotic variance simplifies to 
\bea
\label{eq:simp.var}
\frac{1}{n^*}\V_{[n^*]}(\hat{\beta}_{\infty,n^*})&=&\beta^2 \frac{1}{n^*} \frac{\sigma^2_{\eta^*}}{\int w(\s)^2 dG(\s)}.
\eea
The $r-2$ term in (\ref{eq:simp.bias}) illustrates the fact that the bias is away from the null in the case of a one-dimensional exposure model
and that more typically it is toward the null and becomes larger with higher-dimensional exposure models,
for a given true exposure surface. This is what occurs empirically in our simulations and examples.

In addition, the ratio of the squared bias to the variance is 
\be
\frac{(r-2)^2}{n^*}\frac{\sigma^2_{\eta^*}}{\int w(\s)^2 dG(\s)},
\ee
which demonstrates that the importance of the bias depends on the dimensionality of the exposure model relative to the sample size and the ratio of the noise to the signal in the exposure data.

\section{Measurement error correction
\label{se:correction}}

We correct for measurement error by means of an optional asymptotic bias correction based on (\ref{eq:cl.bias}) followed by a design-based nonparametric bootstrap standard error calculation (incorporating the asymptotic bias correction in the bootstrap, if appropriate).  

Given a bias estimate $\hat{b}$ from  (\ref{eq:cl.bias}) the bias-corrected $\hat{\beta}_{bc}$ is $\hat{\beta}/(1+\hat{b})$.  Bias correction is optional since the asymptotic results of Section~\ref{se:decomp} 
show that the naive health effect estimator is consistent, with variance dominating the bias in the limit as the number of exposure observations increases. 
We explore the magnitude of bias and utility of including the asymptotic correction via simulation in the next section and comment further on this topic in 
Section~\ref{se:disc}.

We need to estimate the uncertainty in either $\hat{\beta}$ or $\hat{\beta}_{bc}$ in a way that accounts for all the components of the measurement error and the sampling variability in the health model. Note that the asymptotic variance (\ref{eq:cl.var}) accounts only for the variance from the classical-like measurement error.
Since we have assumed that the locations of health and exposure data are randomly drawn according to the densities $g(\s)$ and $h(\s)$, respectively, a simple design-based nonparametric bootstrap is a suitable approximation to the data-generating mechanism. 
To obtain each bootstrap dataset, we separately resample with replacement $n^*$ exposure measurements and $n$ health observations.  We fit the
exposure model to the bootstrapped exposure measurements and use the results to predict exposures at the locations of the bootstrapped health observations.  We then obtain
bootstrap health effect estimates (with or without bias correction) and estimate the standard error of $\hat{\beta}$ or $\hat{\beta}_{bc}$
by means of the empirical standard deviation of these values.

In principle, we could avoid the asymptotic calculations in (\ref{eq:cl.bias}) by employing a bootstrap procedure to estimate bias followed by a second round of bootstrapping for standard error estimation.
Such a nested bootstrap is computationally demanding. Furthermore, our strategy of using the bootstrap after addressing the bias is consistent with the comments of \cite[Chapter 10]{Efron1993} who caution that bias correction with the bootstrap is more difficult than variance estimation. Along similar lines, \cite[p. 216]{Buon2010} notes the need for additional assumptions when developing a two-stage bootstrap that includes bias correction.  

\section{Simulations\label{se:sims}}

\subsection{One dimensional exposure surface}
\label{se:1d.ex}
Our first set of simulations is in the simplified setting of a one dimensional exposure surface.  In this setting, we illustrate 
the bias from Berkson-like error for very large $n^*$ when either Condition~\ref{cond1} or~\ref{cond2} is violated, and we illustrate the finite $n^*$ measurement
error correction methods from Section~\ref{se:correction} when both compatibility conditions are satisfied. 
We simulate 1,000 Monte Carlo datasets
and use 100 bootstrap samples, where applicable.  

The true health model is linear regression with $\beta=1$, with i.i.d. $\epsilon \sim N(0,1)$ and an intercept but no additional health model covariates.  We use $n=500$ subjects.
The true exposure surface on $(0,10)$  is a combination of low frequency and high frequency sinusoidal components 
\[
\Phi(\s) = \sin(\s+3.5) + \frac{\s+4}{20}\sin(4 \s - 10.5),
\]
and we set  $\sigma^2_\eta=\sigma^2_{\eta^*}=0.5$.
The density of monitor locations is 
\be
\label{eq:h.dens}
h(\s)=\left\{
\begin{array}{ll}
0.142&0<\s\leq \frac{10}{3}, \frac{20}{3}<\s<10\\
0.0142& \frac{10}{3}\leq \s \leq \frac{20}{3},
\end{array}
\right.
\ee
and we use an exposure model $\capr(\s)$ comprised of a B-spline basis with $5$ to $25$ degrees of freedom \citep{Hastie2001}. 

To illustrate the bias from the Berkson-like error when either of the compatibility conditions is violated, we set $n^*=1000$ so that the classical-like error is negligible.  The results of these simulations are shown in Figure~\ref{fi:ubl}.  In panels~(a) and~(b), the health model is fit with an intercept but no additional health model covariates, so Condition~\ref{cond2}
is automatically satisfied.  In panel~(a) the density of subject locations $g(\s)$ is the same as $h(\s)$, and there is no evidence of bias in $\hat{\beta}$.
In panel~(b), $g(\s)$ is uniform on the interval $(0,10)$ so that Condition~\ref{cond1} is violated.  There is clear 
evidence of bias away from the null for $5$ and $9$ df exposure models.  There is no evidence of bias with $13$ df,
which can be attributed to the fact that the exposure model with $13$ df is sufficiently rich to account for almost all
of the spatial structure in $\Phi(\s)$, meaning that the Berkson-like error behaves like pure Berkson error.

In panels~(c) and~(d) of Figure~\ref{fi:ubl}, $g(\s)$ is the same as $h(\s)$, but we fit the health model including an additional covariate $z_i=\sin(\s_i)$.  In panel~(c), this
covariate is also included in the exposure model, and as expected we see no evidence of bias in $\hat{\beta}$.
In panel~(d), the additional covariate is not included in the exposure model, so Condition~\ref{cond2} is violated.
There is noticeable bias of $\hat{\beta}$ toward the null, especially for the $5$ and $9$ df spline models. 

In  Figure~\ref{fi:corr}, we show results from a separate set of simulations with $n^*=200$ in order to illustrate the measurement
error correction methods from Section~\ref{se:correction}.  In these simulations, $g(\s)$ is the same as $h(\s)$ and the health model is fit without additional covariates, so Conditions~\ref{cond1} and~\ref{cond2} are satisfied.
The mean out-of-sample
$R^2$ ranges from $0.25$ for 5 df to $0.35$ for 13 df, corresponding to the challenging situation of an exposure model with marginal performance that can lead to substantial bias in estimating $\beta$.
In panel~(a), we see that the uncorrected health effect estimates have notable bias, especially for larger df exposure
models, and our correction successfully removes most of the bias.
Panel~(b) shows the coverage of nominal 95\% confidence intervals.
In the uncorrected analyses, coverage ranges from 45\% to 80\%, depending on the df in the exposure model.
Confidence intervals that incorporate either the bias correction or bootstrap standard errors
improve the coverage.  We obtain nearly perfect 95\% coverage when we incorporate the bias correction and bootstrap standard errors.

\subsection{Spatial exposure surface}
\label{se:spatial.ex}

Our second set of simulations is based on the MESA Air study design in the Baltimore region (Section~\ref{se:mesa}), using 1,000 simulated datasets
and 100 bootstrap samples.  
We enforce Conditions~\ref{cond1} and~\ref{cond2} and focus on illustrating the 
value of the correction methods described in Section~\ref{se:correction} in a realistic spatial setting.
The spatial domain is a $257 \times 257$ discrete grid scaled to be a square 30 units on a side.  There are $n^*=125$ monitor locations,
sampled in clusters by first choosing 25 locations i.i.d. uniformly on $S$ and then also including the four nearest 
neighbors for each such location.
Our bootstrap for these simulations resamples clusters of five monitors.
A total of $n=600$ subject locations are selected uniformly and independently from $S$.  

The predictable part of the exposure surface is
\[
\Phi(\s)=\gamma_0+\gamma_1 R_1(\s) +    \gamma_2 R_2(\s) +    \gamma_3 R_3(\s) + \Phi_1(\s).
\]
Each $\gamma_i=4.9$, and each $R_k(\s)$ is constructed by drawing  i.i.d. realizations from $N(0,1/3)$ at each $\s\in S$.
$\Phi_1(\s)$ is a fixed realization from a spectral approximation to a Gaussian field with Mat\'{e}rn covariance \citep{Paciorek2007} with range 20 and unit differentiability parameter, 
normalized such that the variance of $\Phi_1(\s)$ on $S$ is 30.  
Thus the total variance of $\Phi(\s)$ on $S$ is
approximately 54.  In the true exposure surface and monitoring data, there is also a nugget with variance  $\sigma^2_\eta=\sigma^2_{\eta^*}=6$.
We consider two spatial scenarios, corresponding to different fixed realizations of $\Phi_1(\s)$.  These surfaces are shown in Figure~\ref{fi:spatial}.

The spatial exposure model has $\capr(\s)$ comprised of $R_k(\s)$ for $k=1,2,3$ and a thin-plate spline basis 
derived by fitting a GAM from the MGCV package in R  \citep{Wood2006} to the observed monitoring data with fixed degrees of freedom (df).
Thus the spatial basis is actually different for each simulated dataset since it depends on the monitor locations, but we keep the same basis functions for the bootstrap analysis within
each simulation run.  We estimate the standard error of $\hat{\bgamma}$ using a sandwich estimator for clustered data implemented in the R package geepack \citep{Geepack2006}.

The true and fitted health models have an intercept but no additional covariates.  We set $\beta=0.1$ and consider  i.i.d. normally distributed $\epsilon \sim N(0,\sigma^2_\epsilon)$ with
$\sigma^2_\epsilon$ equal to 200 or 10.  The larger value of $\sigma^2_\epsilon=200$ is consistent with what we see in the MESA Air data
with left ventricular mass index (LVMI) as the outcome,
where the air pollution exposure explains approximately $0.3\%$ of the variance after
adjustment for known risk factors.  We also consider $\sigma^2_\epsilon=10$ such that  air pollution exposure explains approximately $5\%$ of the health outcome
variance in order to see more clearly the potential impact of exposure measurement error.


The two spatial surfaces, while generated at random, represent different deterministic scenarios 
in which we could find ourselves (e.g., different metropolitan areas).  In scenario~1 the spatially
structured part of the air pollution surface $\Phi_1(\s)$ can be represented fairly well using thin-plate splines with either 5 or 10 df, while the 
spatial surface in scenario 2 cannot be represented well with 5 df but can be reasonably well modeled with 10 df.  This is reflected in
the $R^2$ values in Figure~\ref{fi:spatial}, which represent the best thin-plate spline fits to the surfaces, assuming essentially unlimited monitoring data is available.
The cross-validated and out-of-sample $R^2$ values for predicting the full air pollution surface $\Phi(\s)$ (including non-spatially structured covariates)  based on monitoring data 
reported in Table~\ref{ta:linear.sim} exhibit a similar pattern.  For leave-one-out cross-validation, the clusters of five adjacent monitors are treated as a single unit.

We focus our discussion on the scenarios with $\sigma^2_\epsilon=10$ because this is where the impact of exposure measurement error is most prominent.  The measurement error impact is 
qualitatively similar for $\sigma^2_\epsilon=200$, but it is less important because the unmodeled variability in the health outcome dominates.  Our theory dictates that the relative biases for $\sigma^2_\epsilon=10$ and $\sigma^2_\epsilon=200$ are identical, which we verified in simulations out to four significant digits, so we only report one value.

When we fit the exposure model with a 5 df thin-plate spline, there is modest bias toward the null of $3\%$ in scenario~1 and more substantial bias of $12\%$ in scenario~2.
Our asymptotic  correction reduces the magnitude of bias in both instances.  The bias correction followed by bootstrap standard errors consistently gives valid inference,
including accurate standard error estimates and nominal coverage of $95\%$ confidence intervals.  In scenario~1, we also get valid inference with bootstrap standard errors
and no bias correction.  

When we increase the complexity of the spatial model to 10 df, prediction accuracy improves in both scenarios, but inference about the health effect
parameter is degraded.  The magnitude of bias is approximately the same as with 5 df, but our asymptotic correction is less effective.   Furthermore, the bootstrap standard error estimates
tend to be too large, resulting in over-coverage of 95\% confidence intervals. 
These findings are not surprising, because while each simulated dataset has $125$ monitor locations, they 
are clustered in groups of $5$ so that there are effectively only $25$ unique locations for estimating the smooth component of the spatial surface, so a thin-plate spline model with 10~df 
overfits these data in the sense that we do not expect to be able to rely on large $n^*$ asymptotic approximations such as (\ref{eq:cl.bias}) or the nonparametric bootstrap. 

\section{Data analysis}
\label{se:mesa}

The Multi-Ethnic Study of Atherosclerosis and Air Pollution (MESA Air) is an ongoing cohort study designed to investigate the relationship between
air pollution exposure and progression of subclinical atherosclerosis \citep{Bild2002,Kaufman2012}.  
The MESA Air cohort includes over 6,000 subjects in six U.S. metropolitan areas (Baltimore City and Baltimore
County, MD; Chicago,
IL; Forsyth County
(Winston-Salem), NC;
Los Angeles and Riverside Counties,
CA; New York and Rockland County, NY;
and St Paul, MN).  Four ethnic/racial groups were targeted, white, African American, Hispanic, and Chinese American,
 and all study participants (46 to 87 years of age) were without clinical cardiovascular disease at the baseline examination
(2000–-2002).
An early cross-sectional finding from MESA Air is that an elevated left-ventricular mass index (LVMI) is associated with exposure to
traffic related air pollution, specifically outdoor residential concentrations of gaseous oxides of nitrogen (NOx) \citep{VanHee2009,VanHee2012}.
\citet{VanHee2012} found 
that an increase in NOx concentration of 10 parts per billion (ppb) is associated with a 0.36 $g/m^2$ increase in LVMI (95\% CI: 0.02 - 0.7 $g/m^2$).

\citet{VanHee2012} utilized predictions from a spatio-temporal exposure model that incorporates regulatory and study-specific monitoring data in all six regions \citep{Szpiro2009}.
To illustrate our methodology for a purely spatial exposure model, we re-analyze the data restricted to subjects in the Baltimore region, and we construct an exposure model based 
on data from three community snapshot monitoring campaigns conducted by MESA Air.  In brief, the community snapshot campaign consisted of
three separate rounds of
spatially rich sampling during single two-week periods in different seasons.  In the Baltimore area, approximately 100
measurements were made in each of three two-week periods in May 2006, November 2006, and February 2007.
In each round of snapshot monitoring, the majority of monitors were arranged in clusters of six, with
three on either side of a major road at distances of approximately 50, 100, and 300 meters \citep{Cohen2009}.  In addition, the locations were chosen
to characterize different land use categories and to cover the geographic region as broadly as possible.  
To help with satisfying Condition~\ref{cond1}, we  exclude one cluster from our analysis
because it is far from any of the study subjects, and we approximate long-term average concentrations by averaging the three available measurements at locations that were monitored in all three seasons.
The 93 monitor locations and 625 subject locations in our analysis are shown in Figure~\ref{fi:mesa}.   

Our exposure model incorporates five geographic covariates: (i) distance to a major road, (ii) local-source traffic pollution from a dispersion model \citep{Wilton2010}, (iii) population density in a 1 km buffer, (iv) distance to downtown, and (v) transportation land use in a 1 km buffer.  The first three of these geographic covariates are log-transformed.
An additional covariate describing the density of high-intensity land-use (commercial, industrial, residential, etc.) was also incorporated in the original spatio-temporal model predictions used by \citep{VanHee2012}, but
we exclude this covariate from our model because it has very different distributions across subject and monitor locations, a clear violation of Condition~\ref{cond1}.
To account for unmodeled spatial structure, we use a thin-plate spline basis with 0, 5, or 10 df, constructed as in the simulations.  We estimate the standard error of $\hat{\bgamma}$ using a sandwich estimator for clustered data implemented in the R package geepack \citep{Geepack2006}.
We estimate the association between NOx and LVMI by fitting a multivariate linear regression, including an exhaustive set of additional 
health model covariates that could be potential confounders \citep{VanHee2009}.

The results of our analysis are shown in Table~\ref{ta:mesa}, with 10,000 bootstrap replicates (resampling clusters of monitors, where applicable). 
Our findings are very similar for an exposure model that is purely land-use regression and one that includes splines with 5 df.  We estimate that an 
increase in NOx concentration of 10 parts per billion (ppb) is associated with approximately a 0.7 $g/m^2$ increase in LVMI.  Our standard error estimates for these models in Table~\ref{ta:mesa} range from $0.55$ to $0.68$ $g/m^2$, so the difference
in effect size from that found by \citet{VanHee2012} is very likely due to our more limited dataset.   
The exposure model that includes 5 df splines
has a larger cross-validated $R^2$, suggesting that it captures more variability in the exposure.  This translates into a smaller model-based standard error, but
this apparent advantage is attenuated when we correct for exposure measurement error with bootstrap standard error estimates, and it goes away entirely
when we also incorporate the bias correction.  

The exposure model with 10 df gives slightly larger effect estimates and standard errors.  There is also more evidence of bias from classical-like error than for the lower
dimensional exposure models.  However, our simulation results in Table~\ref{ta:linear.sim} suggest that a 10 df spline is too rich of a model for the available monitoring data and that these results should be considered less reliable than those based on 5 df splines.

\
\section{Discussion}
\label{se:disc}

We have developed a statistical framework for characterizing and correcting measurement error in two-stage analyses, focusing particularly
on problems where a first-stage spatial model is used to predict exposure that is measured at different locations than are needed in a second-stage health analysis. Our methodology is robust to misspecification of the exposure model, treating it as a device to explain some portion of the variability in exposure. We adopt a design-based perspective in which the process of selecting exposure measurement and subject locations is the primary source of spatial randomness, leading naturally to nonparametric bootstrap resampling for standard errors.  A major contribution of our work is that we delineate the 
potential sources of bias from Berkson-like and classical-like measurement error and provide strategies for reducing bias and variance at the design and analysis stages.  Bias from classical-like error can be corrected using an asymptotic approximation, whereas bias from Berkson-like error 
should be addressed at the design stage or when selecting an exposure model.

While our research is primarily motivated by epidemiologic analysis of long-term air pollution health effects, we note that the spatial prediction problem 
can be interpreted as a linear model. Thus, our measurement error decomposition, asymptotic results, and bias correction hold equally well in non-spatial settings. 

Our theory and simulations demonstrate that bias from the classical-like error is small when the exposure model is not overfit in
the sense that there are sufficient observations relative to the dimension of the exposure model for the large $n^*$ asymptotics to be relevant. The limited magnitude of the bias suggests that measurement error correction efforts should focus on avoiding overfitting the exposure model and satisfying the conditions needed to ensure that Berkson-like error does not induce important bias (at least in a linear health model). 
Nonetheless, in several simulation scenarios our asymptotic correction for bias from classical-like error 
results in improved estimation and inference, even at the expense of increased variance. Indeed, in our analyses and simulations the increased variance caused by estimating the bias is modest.  

Our theoretical development motivates the use of a nonparametric bootstrap to account for variability induced by measurement error. When the bias correction is not used, simulations suggest that the underestimation of uncertainty from ignoring the measurement error (using a sandwich variance estimator) is modest, but even so there are cases in which accounting for the effect of measurement error is necessary. 
When we include the asymptotic bias correction, the bootstrap is more generally necessary for valid confidence intervals. 

As we remarked in Section~\ref{se:compatibility}, exposure model selection is a broad topic and a specific algorithm for selecting geographic 
covariates or spline basis functions is beyond the scope of this paper.
However, we discuss below several practical approaches that can be considered in designing a study to approximately satisfy the compatibility conditions from Section~\ref{se:compatibility}, so as to minimize the bias from Berkson-like error.
We will explore these options and related tradeoffs further in future work.

First, to satisfy Condition~\ref{cond1}, as much care as
possible should be taken at the design stage to ensure the sampling
densities of locations and exposure covariates are as similar as
possible in the first-stage exposure observations and the second-stage outcome
observations. While this criterion is overly abstract in the context of a
specific study, the practical implication is that first-stage and
second-stage locations should be chosen to be similar in terms of
location and pertinent covariates. If exposure data have already been
collected, it may be necessary to consider excluding exposure
or outcome data or deleting one or more covariates from $\capr(\s)$ in order to minimize the mismatch. 

If we are particularly concerned about Condition~\ref{cond2}, we can add terms to $\capr(\s)$ to span $\bTheta(\s)$.  We
generally will not know $\bTheta(\s)$ directly, but if we do (e.g., if
household income were known and monitors were located at homes) then
supplementing $\capr(\s)$ with $\bTheta(\s)$ or projecting $\capr(\s)$
to make it orthogonal to $\bTheta(\s)$ are equivalent. 
In most
realistic settings, we will assume that $\bTheta(\s)$ is a set of smooth
functions of space that can be modeled by spline terms, but we will not
know the minimal spanning spline basis.  In this case it is preferable
to supplement $\capr(\s)$ with as rich of a basis as possible without
introducing substantial classical-like error.
Projecting $\capr(\s)$ to make it orthogonal to a similarly rich
spline basis would likely result in a significant diminution of
exposure variability beyond what is needed to eliminate bias from not satisfying Condition \ref{cond2}. 

The possibility of adding dimensions to $\capr(\s)$ highlights the critical tradeoff between Berkson-like and classical-like error. 
Augmenting $\capr(\s)$ reduces Berkson-like error by accounting for more of the
variability in $w(\s)$. Since eliminating Berkson-like error also
eliminates the need to satisfy the
compatibility conditions, we generally
expect that adding such terms will limit bias from the Berkson-like
error. A side effect of augmenting $\capr(\s)$ is to change the
sampling variability of $\hat{\bgamma}$, which impacts the
classical-like error.  This could be beneficial if the additional
terms in $\capr(\s)$ account for a substantial amount of variability
in $w(\s)$, since the result will be to reduce the variance of the
original components of $\hat{\bgamma}$.  On the other hand, if the
coefficients for the new terms are difficult to estimate, the result will be a substantial
new contribution to the classical-like error, leading to additional
bias and variance in the second-stage estimation.  In fact, in order to reduce classical-like error, one might choose to
remove selected dimensions from $\capr(\s)$ if their coefficients are
particularly difficult to estimate.  

There are some key assumptions in our model that may not be strictly
satisfied in air pollution epidemiology studies. First, we regard the
sets of locations of exposure and health observations to be
independent, or at least independent clusters. This assumption can be
questioned, particularly in the case of air pollution monitors, as one
would not expect a government agency to select two sites that are very
close together. Second, a major source of exposure heterogeneity that
we do not consider is the difference between exposure at a residence
and the exposure experienced by individuals when they are not
home. Mobility may be less important in studies of small children and
the elderly, but this remains an open issue in the epidemiologic
literature. 

Finally, as described in Section~\ref{se:dgm}, we condition
on the unobserved but deterministic spatial variation in exposure during the time period of
the study. This avoids having to postulate
that one could meaningfully repeat the experiment in other time
periods.  This is particularly important when the averaging period of interest is one or more years since
secular trends in the nature and sources of air pollution limit the number of years during which air pollution 
studies can be regarded as answering analogous scientific questions. 
For shorter-term studies, there
is additional variability associated with the choice of time period, and it would be reasonable to regard
the different air pollution surfaces at different times as arising from a random spatial process. 
However, with data from only a single time period and a misspecified mean model, it is impossible to identify both the fixed and random components of the spatial residuals, 
so we do not incorporate
a random effect in our formulation.

Our measurement error correction is based on asymptotic approximations
derived for linear regression for the exposure and health
models. Real world applications often involve additional
complications, suggesting further research directions. On the
exposure model side, our methods can be extended to penalized models and full-rank models such as
universal kriging and related spatio-temporal models that are often used in environmental studies. 
Nonlinear models such as logistic regression and Cox regression are
commonly used for the second stage in health studies, and it is also important to consider the implications of
misspecification in the second-stage model, in addition to the exposure model. 

Two-stage analyses to date have taken the approach of optimizing the
exposure model for exposure prediction accuracy, based on the implicit
assumption that this will also lead to optimal second-stage health
effect inference.  In previous work we have shown that optimizing the
exposure model for prediction accuracy can be sub-optimal for health
effects estimation \citep{Szpiro2011epi}. An interesting avenue for
future research involves developing methods to optimize the exposure
model for estimation of the health effect of interest in the
second-stage model.  A final direction for additional research that is
of great interest in air pollution epidemiology is to extend these
methods for measurement error correction when assessing health effects
of multiple exposures or mixtures of exposures. When predictions for more than one
exposure are used in a health model, there is the possibility of a
form of omitted variable bias from components of variability that are
missing from the predictions of the exposures.

\section*{Acknowledgments}

AAS was supported by the United States Environmental Protection Agency through R831697 and RD-83479601 
and by the National Institute of Environmental Health Sciences through R01-ES009411 and 5P50ES015915.
CJP was supported
by the National Institute of Environmental Health Sciences through ES017017-01A1.
We thank Brent Coull and Lianne Sheppard for comments on the manuscript and Joel Kaufman, Victor Van Hee, and the MESA Air data team for assistance with MESA Air data.

Although the research described in this presentation has been funded wholly or in part by the United States Environmental
Protection Agency through R831697 and RD-83479601 to the University of Washington, it has not been subjected to the Agency’s required peer
and policy review and therefore does not necessarily reflect the views of the Agency and no official endorsement should be
inferred.
\clearpage
\newpage
\appendix
\section{Proof of of Lemma~\ref{le:ucl}}
\label{se:appendixa}

\begin{proof}[Proof of Lemma~\ref{le:ucl}]
By definition, $\hat{\beta}_{\infty,n^*}$ is the 'true' parameter value in a linear model for $y$ on $\hat{w}(\s)$ and covariates, $\z$, so we can express
\bes
y_i = \hat{\beta}_{\infty,n^*} \hat{w}(\s_i)+\hat{\bbeta}_{z,\infty,n^*}^\top \z_i + \nu_i,
\ees
for some $\hat{\bbeta}_{z,\infty,n^*}$ and an error term, $\nu_i$, that is orthogonal to $\hat{w}(\s_i)$ and to the elements of $\z_i$ with respect to the density $g(\s)$.
We can rewrite this expression as
\be
\label{eq:y.expand}
y_i =  \hat{\beta}_{\infty,n^*} \hat{w}^c(\s_i) +\left( \hat{\bbeta}_{z,\infty,n^*}^\top \z_i + \nu_i + \hat{\beta}_{\infty,n^*}(\hat{w}(\s_i)-\hat{w}^c(\s_i))\right)
\ee
and note that each of the last three terms is orthogonal to $ \hat{w}^c(\s_i)$. Given this orthogonality, we can view the last three terms as a single error term. 
The OLS solution for this no-intercept linear model,  
\be
\label{eq:beta.inf2}
 \frac{\sum_{i=1}^n y_i \hat{w}^c(\s_i)}{\sum_{i=1}^n \hat{w}^c(\s_i)^2},
\ee
converges a.s. to $\hat{\beta}_{\infty,n^*}$ for any fixed $n^*$ based on Lemma 1 of \citet{White1980}. We can express each observation as
\bea
\label{eq:yi}
y_i = \beta w^c(\s_i)+\beta(w(\s_i)-w^c(\s_i))+\beta(\bPhi(\s_i)-w(\s_i))+\beta \eta_i + \bbeta_z^\top \z_i + \epsilon_i
\eea
and plug in for $y_i$ in (\ref{eq:beta.inf2}). Taking the limit as $n\to\infty$ shows that the a.s. limit of (\ref{eq:beta.inf2}) can be expressed as
\bea
\label{eq:beta.inf1}
\hat{\beta}_{\infty,n^*} = \beta\frac{ \int w^c(\s) \hat{w}^c(\s) dG(\s)}{\int \hat{w}^c(\s)^2 dG(\s)},
\eea
where the limiting value on the right hand side is found by dividing the numerator and denominator of (\ref{eq:beta.inf2}) by $n$ and invoking the law of large numbers.
In the numerator, only the first summand in the expression for $y_i$ gives a non-zero contribution. 
The contribution from the second summand is zero because $(w(\s_i)-w^c(\s_i))$ is a linear combination of elements of 
 $\bTheta(\s)$ while $w^c(\s)$ is orthogonal to each component of $\bTheta(\s)$.  For the third summand, we use our assumption that Condition~\ref{cond2} is satisfied.  Namely, $w^c(\s)$ is a linear combination of elements of $\bTheta(\s)$ and $\capr(\s)$, so it
is sufficient for $(\bPhi(\s_i)-w(\s_i))$
to be orthogonal to each of the elements of $\bTheta(\s)$ and $\capr(\s)$, and this is guaranteed if the elements of $\bTheta(\s)$ are in the span of $\capr(\s)$ as required by Condition~\ref{cond2}.

We now define the $p\times p$ matrix $\capa=\int \capr^c(\s) ^\top \capr^c(\s) dG(\s)$ and set
\bes
f(\hat{\bgamma})= \left( \bgamma^\top \capa \hat{\bgamma}\right) \left(\hat{\bgamma}^\top \capa \hat{\bgamma}\right)^{-1},
\ees
so that $\hat{\beta}_{\infty,n^*} = \beta f(\bgamma)$.
The gradient of $f(\hat{\bgamma})$ is
\bes
Df(\hat{\bgamma}) = -2(\bgamma^\top \capa \hat{\bgamma}) (\hat{\bgamma}^\top \capa \hat{\bgamma})^{-2} \capa \hat{\bgamma} +  (\hat{\bgamma}^\top \capa \hat{\bgamma})^{-1} \capa \bgamma,
\ees
and its Hessian is
\beas
D^2f(\hat{\bgamma}) &=& {}-2(\bgamma^\top \capa \hat{\bgamma}) (\hat{\bgamma}^\top \capa \hat{\bgamma})^{-2} \capa + 8 (\bgamma^\top \capa \hat{\bgamma}) (\hat{\bgamma}^\top \capa \hat{\bgamma})^{-3} \capa \hat{\bgamma} \hat{\bgamma}^\top \capa \\
&& {} - 2 (\hat{\bgamma}^\top \capa \hat{\bgamma})^{-2} \capa \hat{\bgamma}{\bgamma}^\top \capa - 2 (\hat{\bgamma}^\top \capa \hat{\bgamma})^{-2} \capa {\bgamma} \hat{\bgamma}^\top \capa.
\eeas
A second order Taylor expansion of $f(\hat{\bgamma})$ can be written 
\beas
f(\hat{\bgamma})&\approx&1-(\bgamma^\top \capa \bgamma)^{-1} (\hat{\bgamma}-\bgamma)^\top \capa \bgamma -  (\bgamma^\top \capa \bgamma)^{-1}(\hat{\bgamma}-\bgamma)^\top\capa (\hat{\bgamma}-\bgamma)\\
&&{} + 2 (\bgamma^\top \capa \bgamma)^{-2}(\hat{\bgamma}-\bgamma)^\top  \capa \bgamma \bgamma^\top \capa (\hat{\bgamma} - \bgamma)\\
&=&1- \frac {\int \left(\hat{w}^c(\s)-w^c(\s)\right) w^c(\s) dG(\s) }{\int w^c(\s)^2 dG(\s)}-\frac{ \int \left(\hat{w}^c(\s)-w^c(\s)\right)^2dG(\s) }{ \int w^c(\s)^2 dG(\s)}\\
&&{}+2 \frac{ \left( \int \left(\hat{w}^c(\s)-w^c(\s)\right) w^c(\s) dG(\s)  \right)^2}{\left( \int w^c(\s)^2 dG(\s)\right)^{2}}\\
&=&1- \frac {\int \left(\hat{w}^c(\s)-w^c(\s)\right) w^c(\s) dG(\s) }{ \int w^c(\s)^2 dG(\s)}-\frac{ \int \left(\hat{w}^c(\s)-w^c(\s)\right)^2dG(\s) }{\int w^c(\s)^2 dG(\s)}\\
&&{}+2 \frac{ \int \left(\hat{w}^c(\s_1)-w^c(\s_1)\right)  \left(\hat{w}^c(\s_2)-w^c(\s_2)\right) w^c(\s_1)w^c(\s_2) dG(\s_1) dG(\s_2) }{\left( \int w^c(\s) dG(\s)\right)^2}
\eeas
We require sufficient regularity such that the higher order terms converge in distribution to zero sufficiently fast as $n^*\rightarrow \infty$. Equations (\ref{eq:cl.bias}) and (\ref{eq:cl.var}) are then derived by taking the asymptotic expectation and asymptotic variance of the expression above (for the asymptotic variance, we can restrict our attention to the first non-constant term above). 
This requires interchanging asymptotic expectations with respect to $\hat{\bgamma}$ with integrals in $\s$. 
A sufficient condition to do this is that $\capr(\s)$ is bounded as a function of $\s$. The interchange is accomplished by expressing asymptotic expectations as standard expectations of the appropriate limiting distributions of functions of $\sqrt{n^*}(\hat{\bgamma} -\bgamma)$ using Definition~\ref{def:asympt}, invoking the continuous mapping theorem to express these limiting distributions as functions of the limiting distribution of $\sqrt{n^*}(\hat{\bgamma} -\bgamma)$, and then using the Fubini-Tonelli theorem \citep{folland1999real} to exchange the order of integration between $\hat{\bgamma}$ and of $\s$. 
There is an additional step to express asymptotic expectations as the asymptotic covariances in (\ref{eq:cl.bias}) and (\ref{eq:cl.var}).
We use Definition~\ref{def:asympt} to express the asymptotic expectations as standard expectations of the appropriate limiting distributions of functions of $\sqrt{n^*}(\hat{\bgamma} -\bgamma)$, again invoke the continuous mapping theorem, express the resulting expectations as covariances of the limiting distributions of functions of $\sqrt{n^*}(\hat{\bgamma} -\bgamma)$,
and finally write the resulting expression as asymptotic covariances using Definition~\ref{def:asympt}.
\end{proof}

\section{Estimating the asymptotic bias} 
\label{se:appendixb}

To estimate the asymptotic bias based on (\ref{eq:kappa}), we require the second derivatives, $\frac{\partial^2 \kappa}{\partial m_j \partial m_k}$.
Differentiating with respect to $m_j$ and $m_k$ gives us a long expression involving first and second derivatives of $M$:
\beas
\rtmri \rtdmkr \rtmri \rtdmjr \rtmri \rtmp - \\\rtmri \rtddmr \rtmri \rtmp + \\
\rtmri \rtdmjr \rtmri \rtdmkr \rtmri \rtmp - \\\rtmri \rtdmjr \rtmri \rtdmkp - \\
\rtmri \rtdmkr \rtmri \rtdmjp + \\\rtmri \rtddmp.
\eeas
The first derivative with respect to $m_j$ is a matrix of zeros with a single one in the $j$th diagonal position, while the second partial derivative is a matrix of zeros. The terms that remain in the expression can be easily calculated based on the observed $\capr^*$ and using the plug-in estimate, $\capx^* = (x_1^*,\ldots,x_{n^*}^*)^\top$, for $\bPhi^*$. Note that $\capr^{*\top} \frac{\partial \capm}{\partial m_j}  \capr^* = \capr_j^* \capr_j^{*\top}$, where $\capr_j^*$ is the $j$th row of $\capr^*$, and $\capr^{*\top} \frac{\partial \capm}{\partial m_j} \capx^* = \capr_j^{*\top} x_j^*$.

\clearpage
\newpage

\bibliographystyle{plainnat}
\bibliography{me}

\clearpage
\newpage
\begin{figure}[pc!]
\centering
\includegraphics[,width=7in]{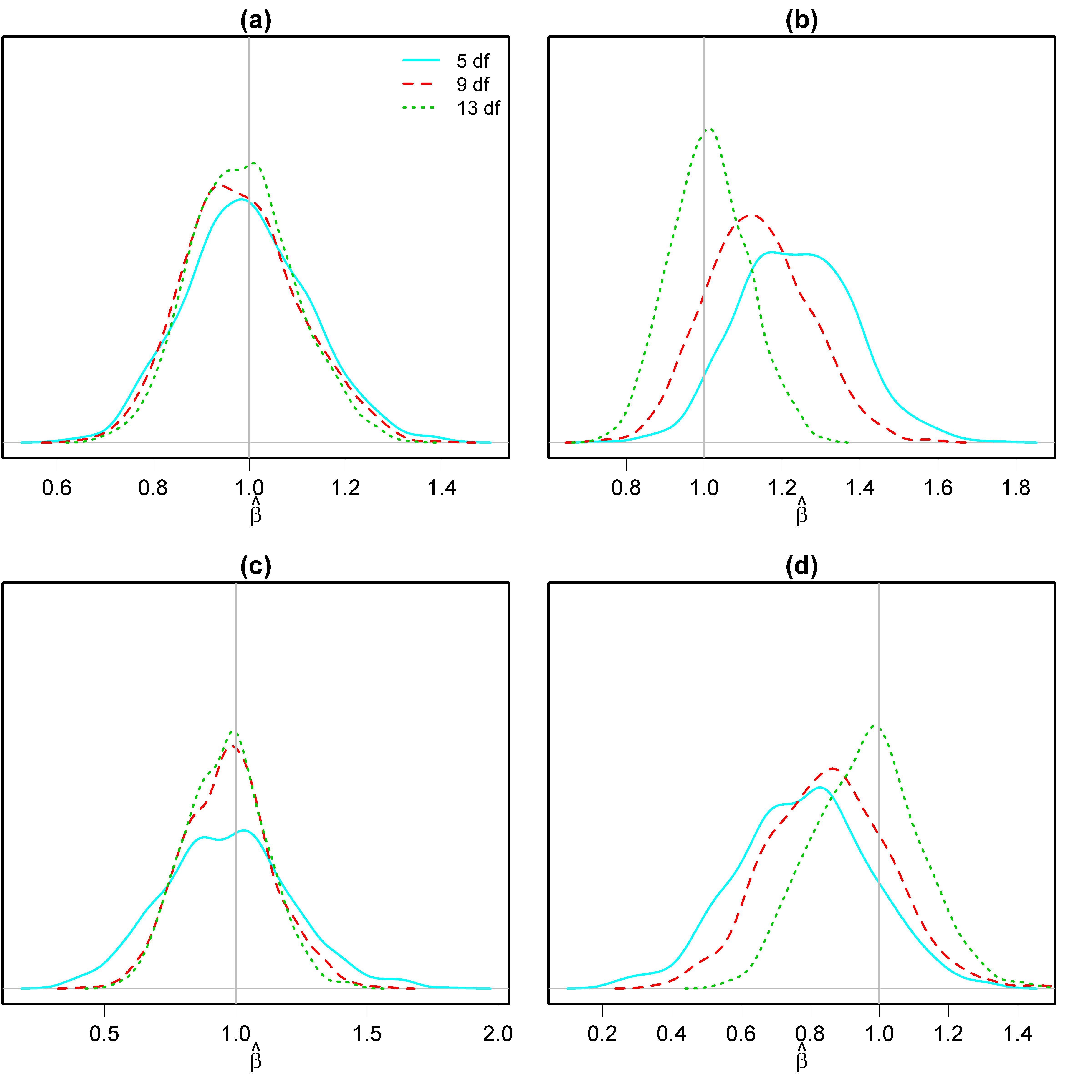}
\caption{Density plots of $\hat{\beta}$ show bias from $u_{i,BL}$ in the one dimensional simulation scenario with $n^*=1000$ when the compatibility conditions are not satisfied (1,000 Monte Carlo simulations for each scenario).  (a) Conditions~1 and~2 are satisfied: subject and monitor locations have the same 
density, and the health model is fit without subject specific covariates.  (b)  Condition 1 is violated: same as (a) except that $g(\s)$ is uniform on $(0,10)$.
(c) Conditions~1 and~2 are satisfied: subject and monitor locations have the same 
density, the health model is fit with a sinusoidal covariate, and the additional covariate is included in the exposure model.  
(d)  Condition 2 is violated: same as (c) except that the sinusoidal covariate is not included in the exposure model.} 
\label{fi:ubl}
\end{figure}

\clearpage
\newpage
\begin{figure}[pc!]
\centering
\includegraphics[,width=7in]{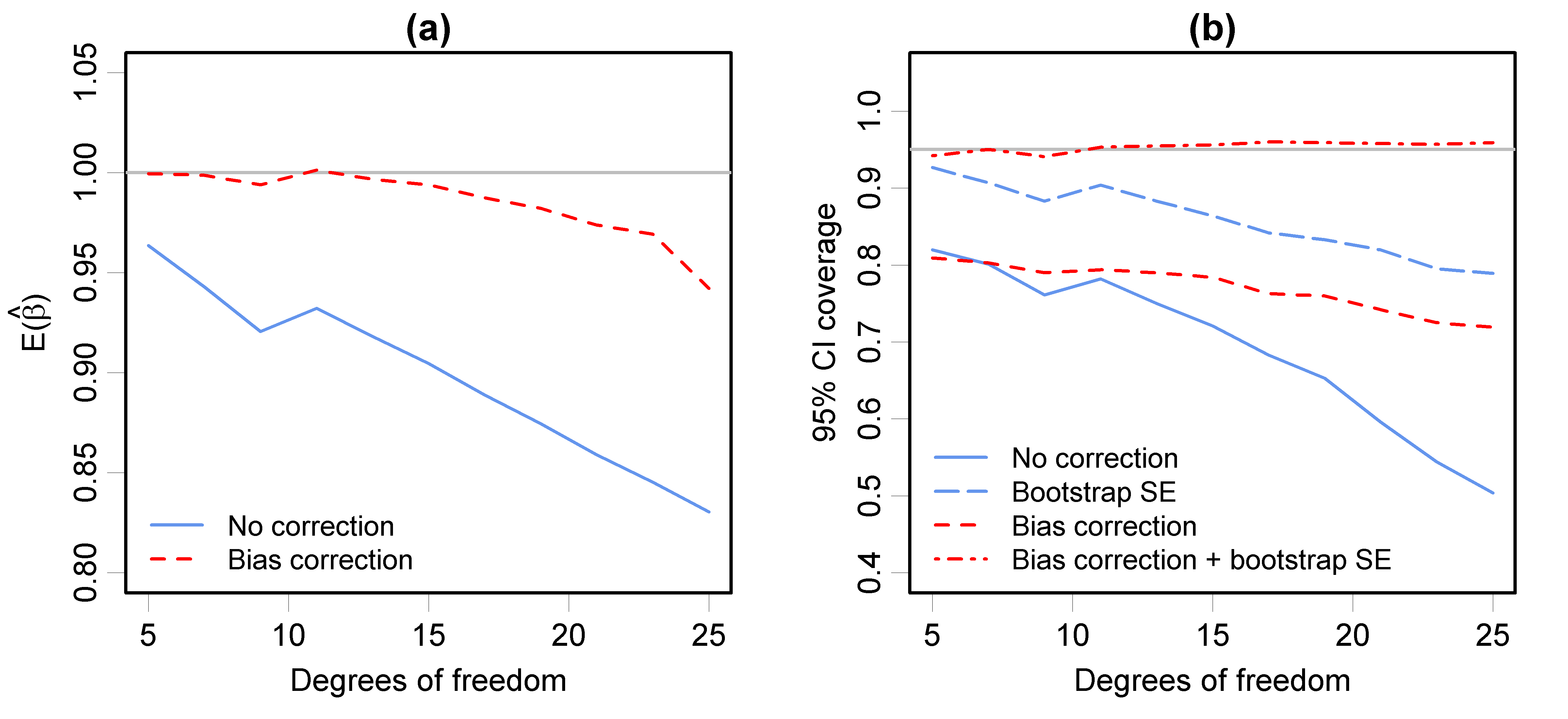}
\caption{ The two-step measurement error correction method adjusts for bias and gives valid standard error estimates
 in the one dimensional simulation scenario with $n^*=200$ (1,000 Monte Carlo simulations for each scenario and 100 bootstrap samples).  (a) Uncorrected health effect estimates have noticeable bias, especially for high dimensional exposure models.  Our analytical correction successfully adjust for the bias.  (b) The combination of bias adjustment and bootstrap standard errors gives 95\% confidence intervals with nominal coverage properties.  Neither bias adjustment alone nor bootstrap standard errors alone are sufficient.} 
\label{fi:corr}
\end{figure}

\clearpage
\newpage
\begin{figure}[pc!]
\centering
\[
\begin{array}{ccc}
\textrm{Scenario 1}&\textrm{5 df }(R^2=0.78)&\textrm{10 df }(R^2=0.91)\\
\includegraphics[,width=2.in]{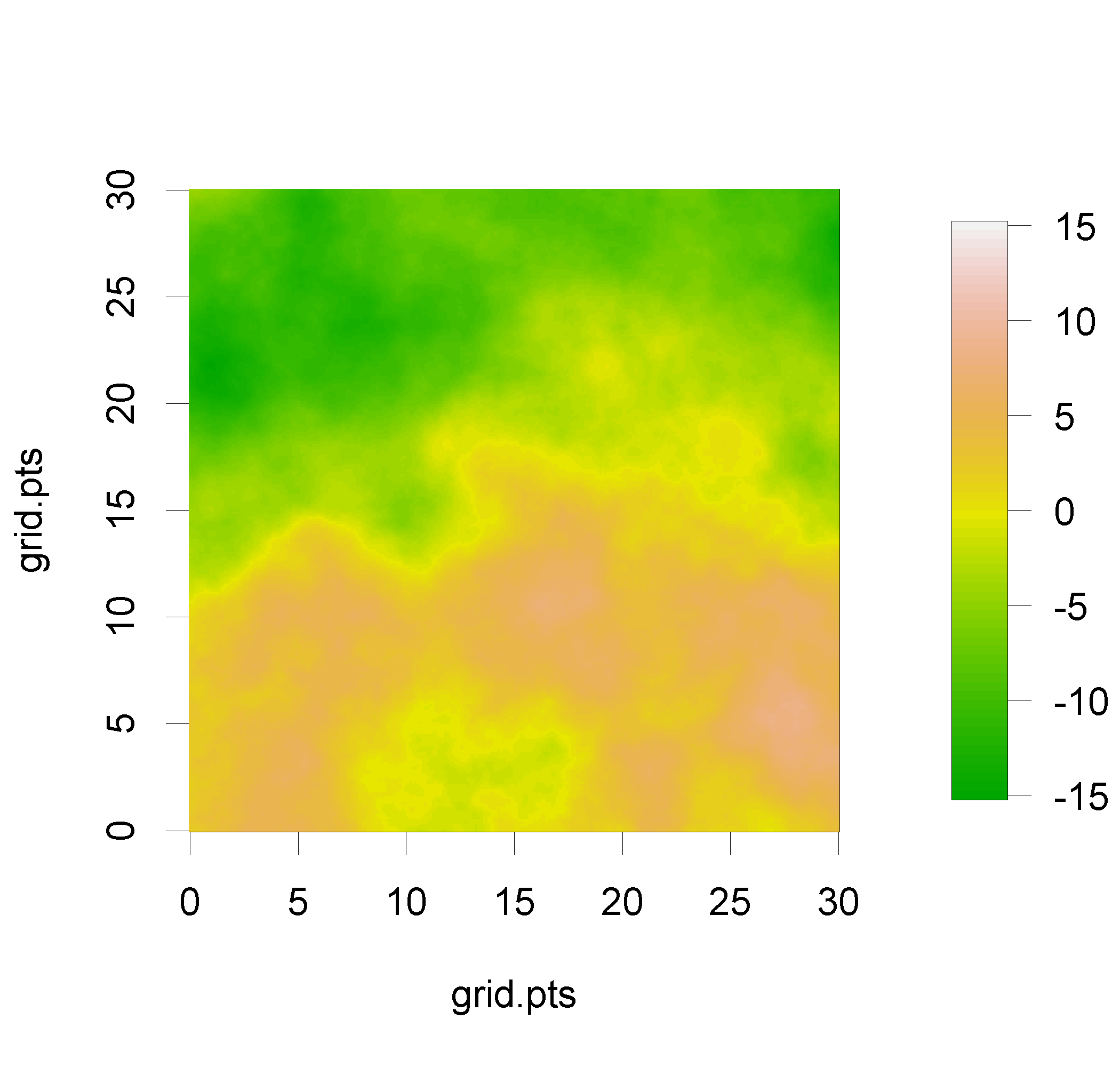}&
\includegraphics[,width=2.in]{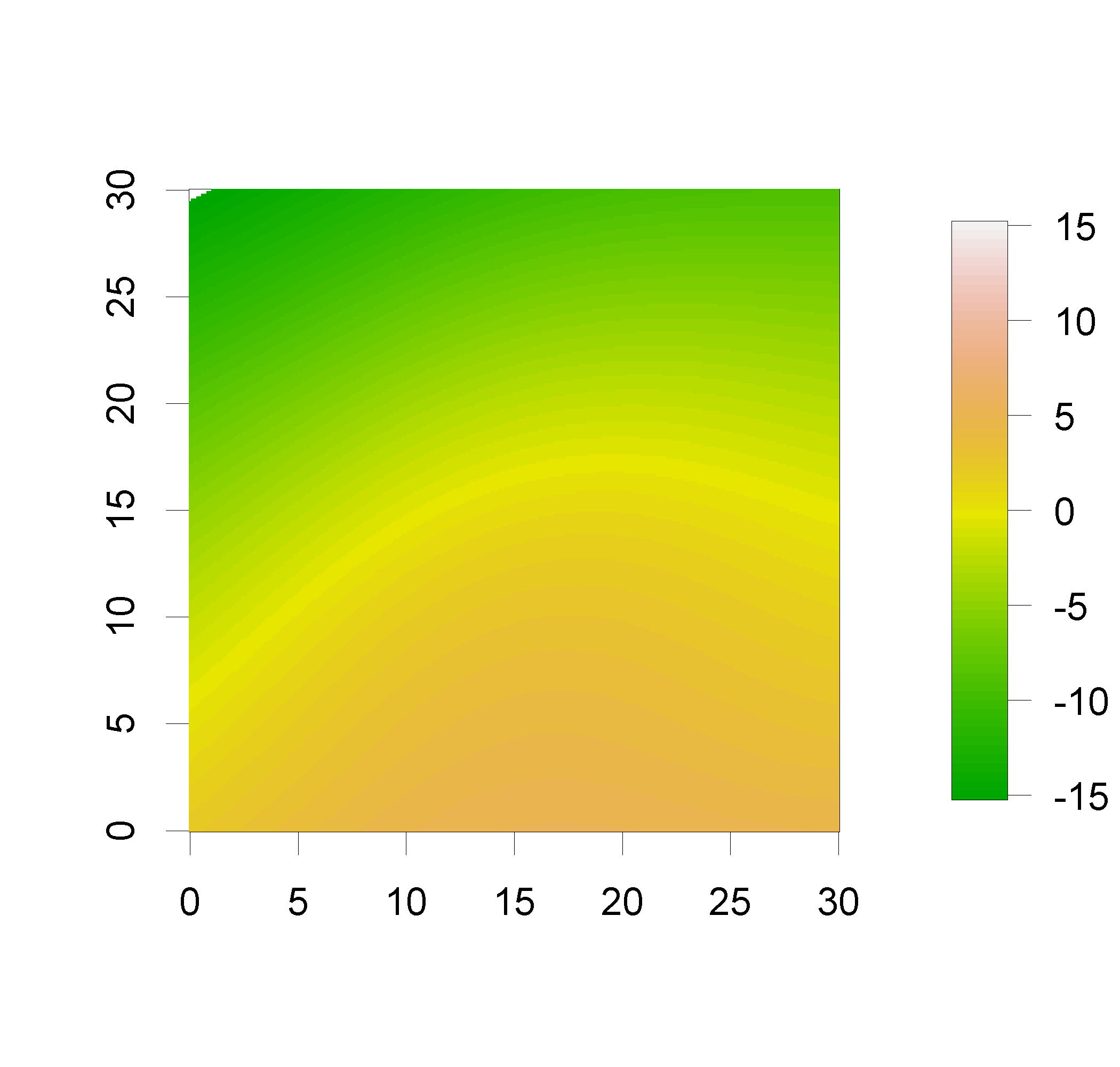}&
\includegraphics[,width=2.in]{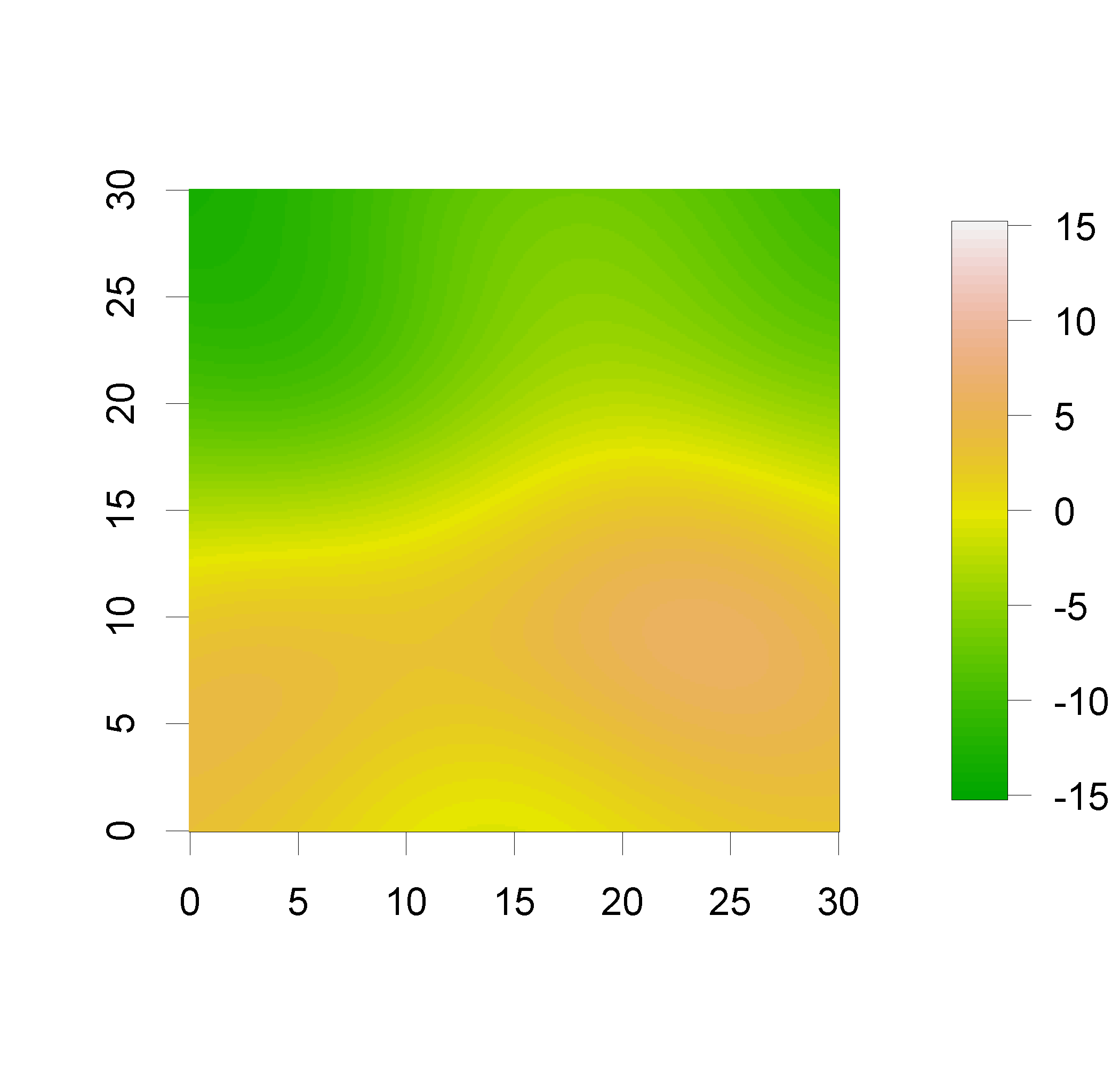}\\
\textrm{Scenario 2}&\textrm{5 df }(R^2=0.22)&\textrm{10 df }(R^2=0.63)\\
\includegraphics[,width=2.in]{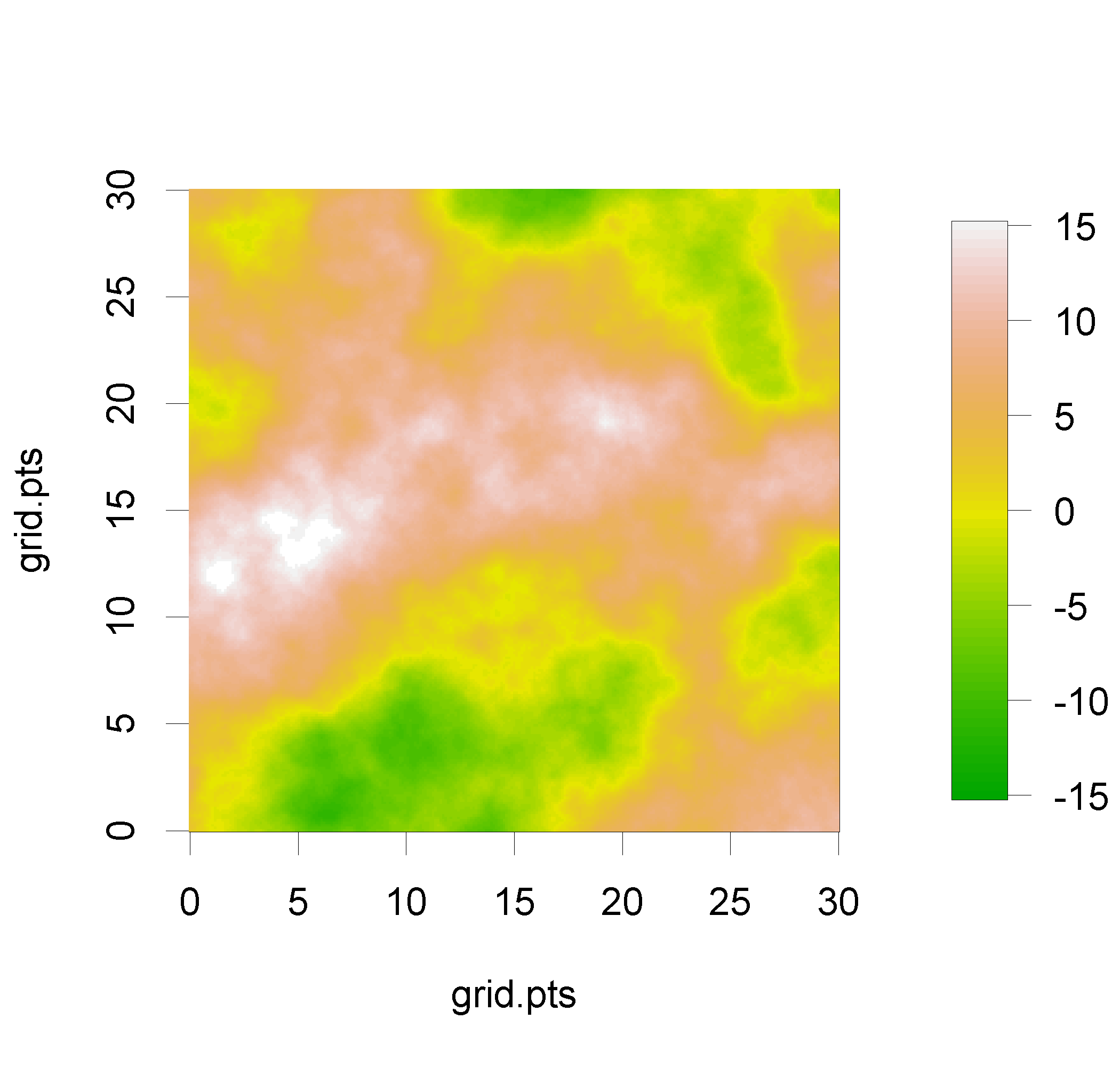}&
\includegraphics[,width=2.in]{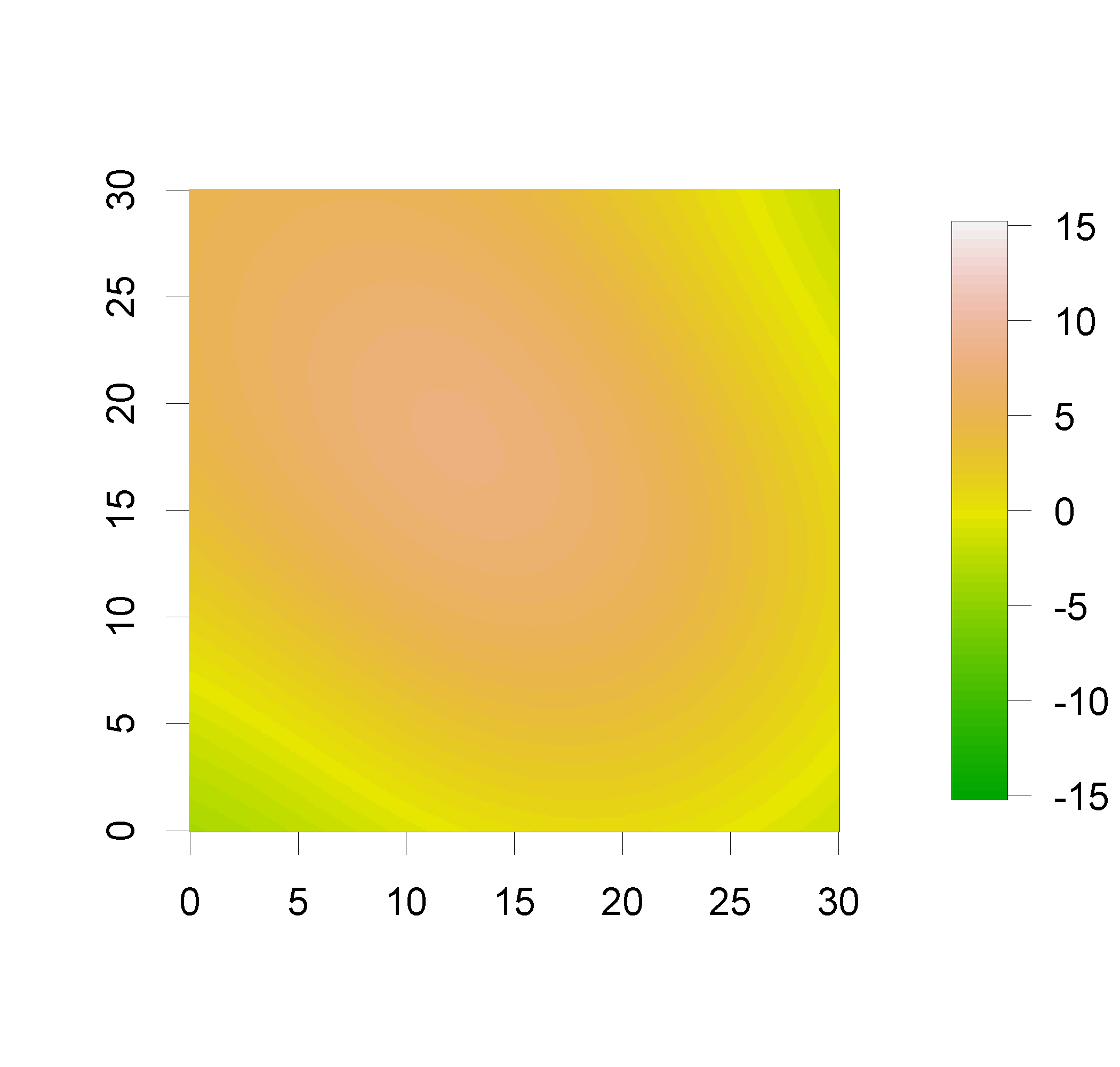}&
\includegraphics[,width=2.in]{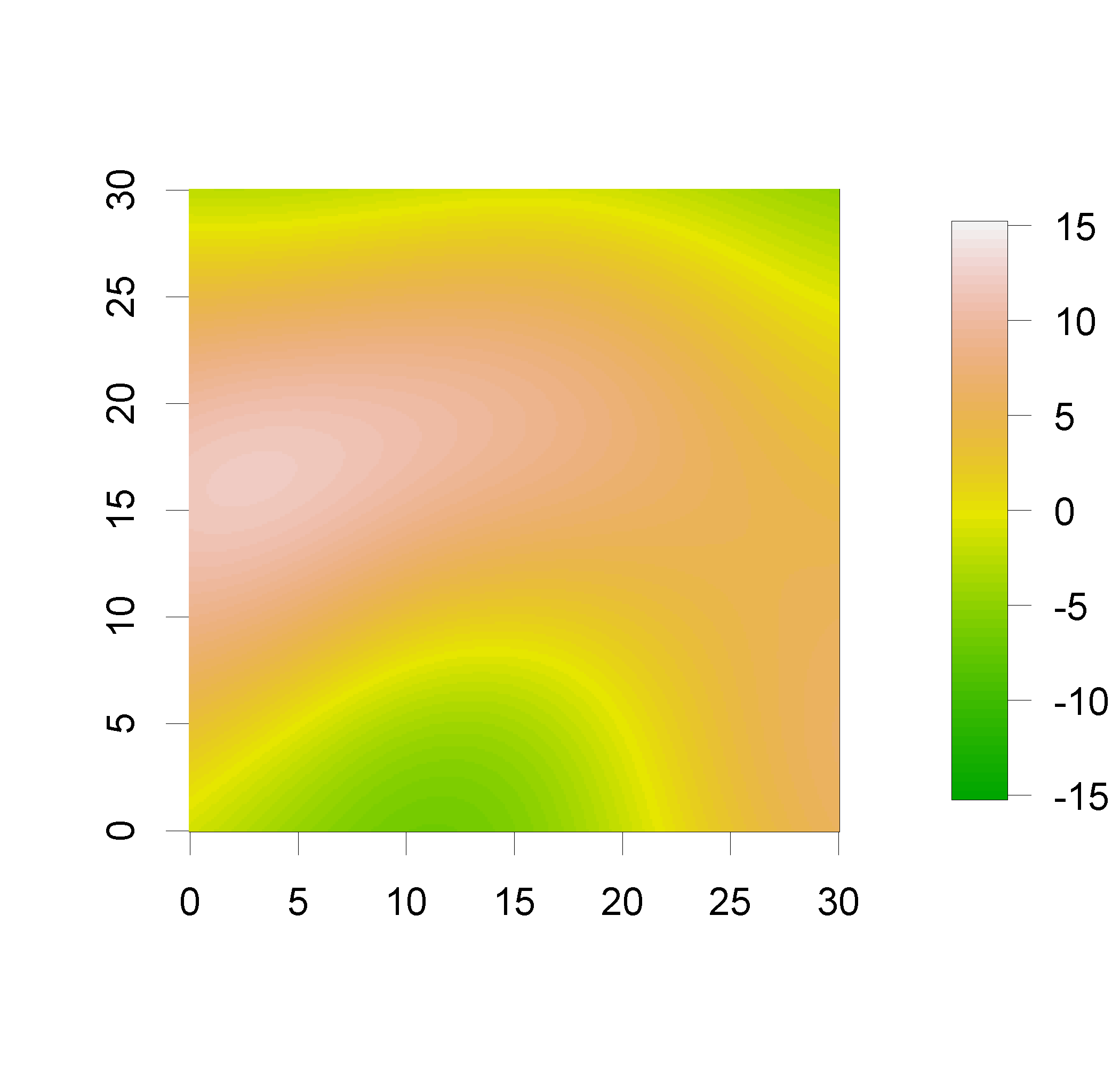}\\
\end{array}
\]
\caption{Spatial component of the exposure surface for each of two scenarios in the spatial simulation study based on MESA Air.
Each scenario corresponds to a different random draw of $\Phi_1(\cdot)$ from a \matern-based Gaussian process with range 20, differentiability
parameter 1, and variance on $S\subset \R^2$ of 30.  The first column is the true surface, and the second and third columns show
prediction surfaces and $R^2$ from approximating $\Phi_1(\s)$
with thin-plate splines using the indicated degrees of freedom, based on fitting a fixed degree-of-freedom spatial GAM model to  $\Phi_1(\s)$ with observations at every location on the 
$257 \times 257$ grid $S$.} 
\label{fi:spatial}
\end{figure}

\clearpage
\newpage
\begin{figure}[pc!]
\centering
\includegraphics[,width=7in]{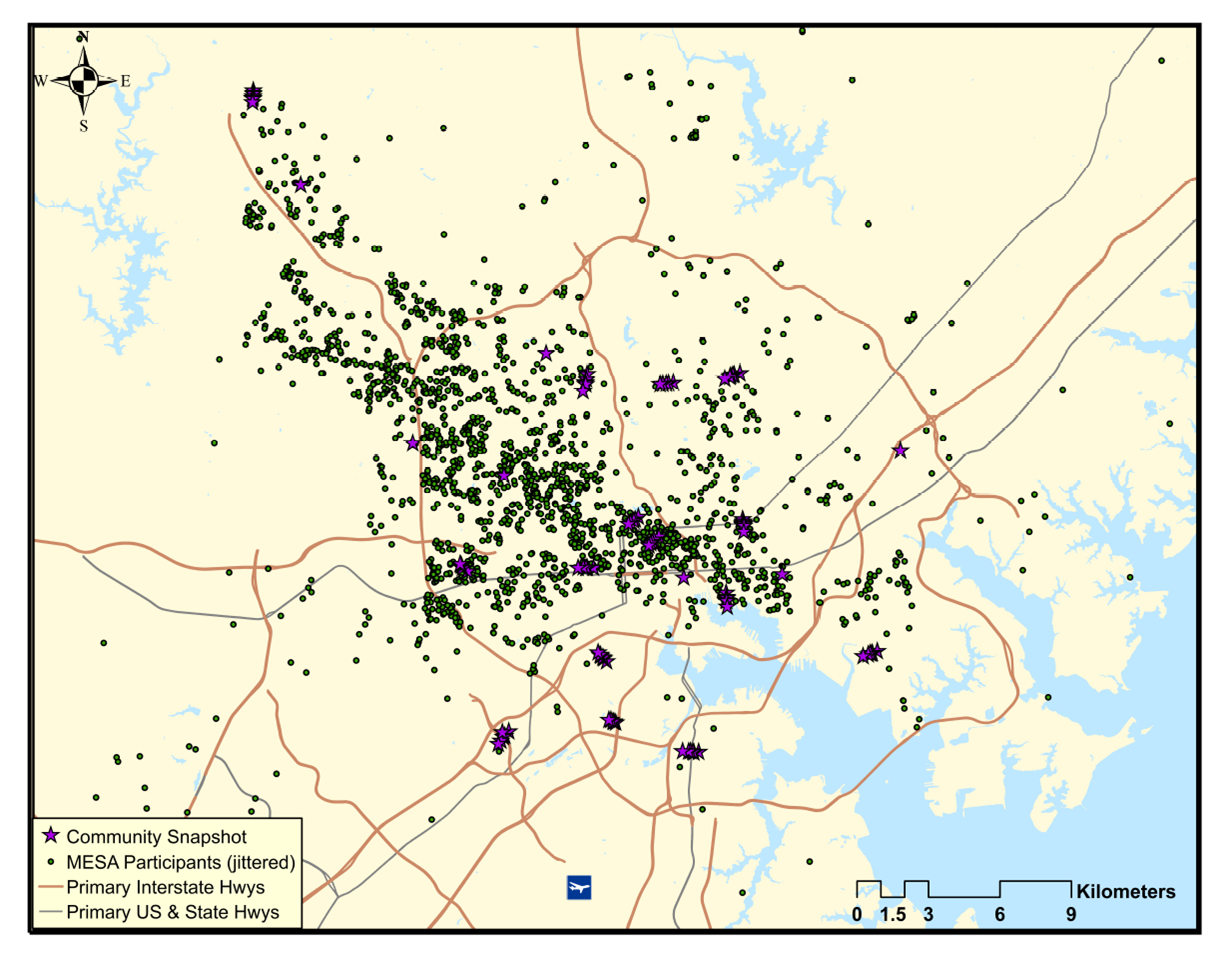}
\caption{Subject and monitor locations in the Baltimore MESA Air dataset.} 
\label{fi:mesa}
\end{figure}

\clearpage
\newpage
\begin{table}
\caption{Results for spatial simulation (1,000 Monte Carlo simulations for each scenario and 100 bootstrap samples).  
The average
out-of-sample $R^2$ is given in parentheses for each exposure model.  The first column is the relative bias in estimating $\beta=0.1$.  
This is the same for $\sigma^2=200$ and $\sigma^2=10$ and is estimated from 100,000 Monte Carlo
samples, resulting in negligible Monte Carlo error.
The final six columns show the standard deviation, average estimated standard error, and 95$\%$ confidence interval coverage, separately for 
$\sigma_\epsilon^2=200$ and $\sigma_\epsilon^2=10$.}
\small
\centering
\[
\ba{lrrrrrrr}
\noalign {\vspace{0.05cm}}
&&\multicolumn{3}{c}{\sigma^2_\epsilon=200}&\multicolumn{3}{c}{\sigma^2_\epsilon=10}\\
&\textrm{Rel Bias}&\textrm{SD}&\textrm{E(SE)}&\textrm{Cov}&\textrm{SD}&\textrm{E(SE)}&\textrm{Cov}\\
\hline
\noalign {\vspace{0.05cm}}
\textrm{Scenario }1\\
\hline
\noalign {\vspace{0.05cm}}
\textrm{5 degrees of freedom }(0.75)\\
\noalign {\vspace{0.05cm}}
\quad \textrm{no correction}&-0.027&0.084&0.083&94\%&0.02&0.019&93\%\\
\quad \textrm{bootstrap standard error only}&-0.027&0.084&0.084&95\%&0.02&0.021&95\%\\
\quad \textrm{bias correction only}&-0.009&0.086&0.083&94\%&0.021&0.019&93\%\\
\quad \textrm{bias correction + bootstrap}&-0.009&0.086&0.086&95\%&0.021&0.021&96\%\\
\noalign {\vspace{0.05cm}}
\textrm{10 degrees of freedom }(0.79)\\
\noalign {\vspace{0.05cm}}
\quad \textrm{no correction}&-0.039&0.08&0.08&95\%&0.019&0.018&93\%\\
\quad \textrm{bootstrap standard error only}&-0.039&0.08&0.082&96\%&0.019&0.027&98\%\\
\quad \textrm{bias correction only}&-0.025&0.081&0.08&94\%&0.019&0.018&93\%\\
\quad \textrm{bias correction + bootstrap}&-0.025&0.081&0.088&97\%&0.019&0.03&99\%\\
\hline
\noalign {\vspace{0.05cm}}
\textrm{Scenario }2\\
\hline
\noalign {\vspace{0.05cm}}
\textrm{5 degrees of freedom }(0.42)\\
\noalign {\vspace{0.05cm}}
\quad \textrm{no correction}&-0.125&0.099&0.096&94\%&0.025&0.022&87\%\\
\quad \textrm{bootstrap standard error only}&-0.125&0.099&0.097&95\%&0.025&0.026&90\%\\
\quad \textrm{bias correction only}&-0.049&0.108&0.096&93\%&0.028&0.022&86\%\\
\quad \textrm{bias correction + bootstrap}&-0.049&0.108&0.107&95\%&0.028&0.03&94\%\\
\noalign {\vspace{0.05cm}}
\textrm{10 degrees of freedom }(0.59)\\
\noalign {\vspace{0.05cm}}
\quad \textrm{no correction}&-0.102&0.087&0.085&93\%&0.021&0.019&88\%\\
\quad \textrm{bootstrap standard error only}&-0.102&0.087&0.085&94\%&0.021&0.03&94\%\\
\quad \textrm{bias correction only}&-0.061&0.091&0.085&92\%&0.023&0.019&88\%\\
\quad \textrm{bias correction + bootstrap}&-0.061&0.091&0.094&95\%&0.023&0.036&97\%\\
\ea
\]
\label{ta:linear.sim}
\end{table}

\clearpage
\newpage
\begin{table}[pc!]
\caption{Results of analysis analyzing the association between elevated left ventricular mass index (LVMI) and residential concentrations of NOx in
the MESA Air cohort in Baltimore, based on exposure models with land-use regression plus a thin-plate spline regression with varying degrees of freedom.  
The two columns display estimates and standard errors for the increase in LVMI ($g/m^2$) associated with a 10 ppb increase in NOx concentration.
Cross-validated $R^2$ at snapshot monitor locations for each exposure model are given in parentheses, based on leave-one-out cross-validation
modified to leave out all members of a roadway gradient cluster together. }
\centering
\[
\ba{lrrrr}
\noalign {\vspace{0.05cm}}
&\hat{\beta}&\textrm{SE}\\
\hline
\noalign {\vspace{0.05cm}}
\textrm{Land-use regression only (0.60)}\\
\quad \textrm{no correction}&0.66&0.62&\\
\quad \textrm{bootstrap standard error}&0.66&0.66&\\
\quad \textrm{bias correction + bootstrap}&0.68&0.68&\\
\textrm{5 degrees of freedom (0.68)}\\
\quad \textrm{no correction}&0.68&0.55&\\
\quad \textrm{bootstrap standard error}&0.68&0.62&\\
\quad \textrm{bias correction + bootstrap}&0.69&0.67&\\
\textrm{10 degrees of freedom (0.41)}\\
\quad \textrm{no correction}&0.79&0.69&\\
\quad \textrm{bootstrap standard error}&0.79&0.66&\\
\quad \textrm{bias correction + bootstrap}&0.85&0.78&\\
\ea
\]
 \label{ta:mesa}
\end{table}

\end{document}